\documentclass[12pt,a4paper]{amsart}

\usepackage{amsmath,amsfonts,amssymb}

\usepackage[hmargin=2cm,vmargin=2cm]{geometry}

\usepackage{hyperref}
\usepackage{nameref,zref-xr}                    

\usepackage{amsthm,amscd}
\usepackage{epsfig}
\usepackage{color}
\usepackage[all]{xy}
\usepackage{tikz}
\usepackage{graphics, setspace}
\usepackage{breqn}
\usepackage{amstext}
\usepackage{array}

\newcommand{\PP}{{\mathbb{P}}}

\newcommand{\pt}{\mathrm{pt}}

\newcommand{\p}{{\partial}}

\newcommand{\mbZ}{\mathbb Z}
\newcommand{\mbC}{\mathbb C}

\newcommand{\tG}{\widetilde{G}}

\newcommand{\tU}{{\widetilde{U}}}

\def\d{{\partial}}

\newcommand{\eps}{\varepsilon}

\newcommand{\Coef}{\mathrm{Coef}}

\newcommand{\tv}{\widetilde v}
\renewcommand{\top}{\mathrm{top}}

\newcommand{\tOmega}{\widetilde\Omega}

\newcommand{\tR}{\widetilde{R}}

\newcommand{\tQ}{\widetilde{Q}}

\newcommand{\mcF}{\mathcal{F}}

\renewcommand{\tt}{\widetilde{t}}
\newcommand{\tr}{\mathrm{tr}}
\newcommand{\diag}{\mathrm{diag}}
\newcommand{\mcB}{\mathcal{B}}

\newcommand{\mcL}{\mathcal{L}}
\newcommand{\orig}{\mathrm{orig}}
\newcommand{\mcO}{\mathcal{O}}
\newcommand{\Id}{\mathrm{Id}}
\newcommand{\tD}{\widetilde{D}}
\newcommand{\const}{\mathrm{const}}
\newcommand{\cL}{\mathcal{L}}
\newcommand{\Mat}{\mathrm{Mat}}
\newcommand{\tmu}{\widetilde{\mu}}

\newcommand{\hatt}{\widehat{t}}

\newcommand{\onabla}{\overline{\nabla}}
\newcommand{\tnabla}{\widetilde{\nabla}}
\newcommand{\mcU}{\mathcal{U}}
\newcommand{\txi}{\widetilde{\xi}}
\newcommand{\oR}{\overline{R}}

\newcommand{\ext}{\mathrm{ext}}

\newcommand{\rspin}{\text{$r$-spin}}
\newcommand{\PST}{\mathrm{PST}}

\newtheorem{theorem}{Theorem}[section]
\newtheorem{proposition}[theorem]{Proposition}
\newtheorem{lemma}[theorem]{Lemma}

\newtheorem{definition}[theorem]{Definition}

\newtheorem{remark}[theorem]{Remark}

\usepackage{color}

\def\&{\vspace{-5pt}&}

\numberwithin{equation}{section}

\begin{document}

\title{Open WDVV equations and Virasoro constraints}

\author{Alexey Basalaev}
\address{Skolkovo Institute of Science and Technology, Nobelya Ulitsa 3, Moscow, Russian Federation 121205}
\email{a.basalaev@skoltech.ru}

\author{Alexandr Buryak\textsuperscript{*}}
\address{School of Mathematics, University of Leeds, Leeds, LS2 9JT, United Kingdom}
\email{a.buryak@leeds.ac.uk}

\thanks{\textsuperscript{*} Corresponding author}

\dedicatory{Dedicated to Rafail Kalmanovich Gordin on the occasion of his 70th birthday}

\date{\today}

\begin{abstract}
In their fundamental work, B. Dubrovin and Y. Zhang, generalizing the Virasoro equations for the genus $0$ Gromov--Witten invariants, proved the Virasoro equations for a descendent potential in genus $0$ of an arbitrary conformal Frobenius manifold. More recently, a remarkable system of partial differential equations, called the open WDVV equations, appeared in the work of A. Horev and J. P. Solomon. This system controls the genus $0$ open Gromov--Witten invariants. In our paper, for an arbitrary solution of the open WDVV equations, satisfying a certain homogeneity condition, we construct a descendent potential in genus $0$ and prove an open analog of the Virasoro equations. We also present conjectural open Virasoro equations in all genera and discuss some examples. 
\end{abstract}

\maketitle

\tableofcontents

\section{Introduction}

The {\it WDVV equations}, also called the {\it associativity equations}, is a system of non-linear partial differential equations for one function, depending on a finite number of variables. Let $N\ge 1$ and $\eta=(\eta_{\alpha\beta})$ be an $N\times N$ symmetric non-degenerate matrix with complex coefficients. The WDVV equations is the following system of PDEs for a function $F(t^1,\ldots,t^N)$ defined on some open subset $M\subset\mbC^N$:
\begin{gather}\label{eq:WDVV equations}
\frac{\d^3 F}{\d t^\alpha\d t^\beta\d t^\mu}\eta^{\mu\nu}\frac{\d^3 F}{\d t^\nu\d t^\gamma\d t^\delta}=\frac{\d^3 F}{\d t^\delta\d t^\beta\d t^\mu}\eta^{\mu\nu}\frac{\d^3 F}{\d t^\nu\d t^\gamma\d t^\alpha},\quad 1\le\alpha,\beta,\gamma,\delta\le N,
\end{gather}
where $(\eta^{\alpha\beta}):=\eta^{-1}$ and we use the convention of sum over repeated Greek indices. Suppose that the function~$F$ satisfies the additional assumption $\frac{\d^3 F}{\d t^1\d t^\alpha\d t^\beta}=\eta_{\alpha\beta}$. Then the function~$F$ defines a structure of {\it Frobenius manifold} on~$M$ and is also called the \textit{Frobenius manifold potential}. Such a structure appears in different areas of mathematics, including the singularity theory and curve counting theories in algebraic geometry (Gromov--Witten theory, Fan--Jarvis--Ruan--Witten theory). A systematic study of Frobenius manifolds was first done by B.~Dubrovin~\cite{Dub96,Dub99}. 

\bigskip

Consider formal variables $t^\alpha_p$, $1\le\alpha\le N$, $p\ge 0$, where we identify $t^\alpha_0=t^\alpha$. There is a natural way to associate to the function $F$ a {\it descendent potential} $\mcF$, which is a function of the variables $t^\alpha_p$, such that the difference $\left.\mcF\right|_{t^*_{\ge 1}=0}-F$ is at most quadratic in the variables $t^1,\ldots,t^N$ and the following equations are satisfied:
\begin{gather}\label{eq:TRR for Frob. man.,1}
\frac{\d^3\mcF}{\d t^\alpha_{a+1}t^\beta_b\d t^\gamma_c}=\frac{\d^2\mcF}{\d t^\alpha_a\d t^\mu_0}\eta^{\mu\nu}\frac{\d^3\mcF}{\d t^\nu_0\d t^\beta_b\d t^\gamma_c},\quad 1\le\alpha,\beta,\gamma\le N,\quad a,b,c\ge 0.
\end{gather}
These equations are called the {\it topological recursion relations (TRR)}. In Gromov--Witten theory, where the function $F$ is the generating series of intersection numbers on the moduli space of maps from a Riemann surface of genus $0$ to a target variety, a natural descendent potential $\mcF$ is given by the generating series of intersection numbers with the Chern classes of certain line bundles over the moduli space. Note that the system of equations~\eqref{eq:TRR for Frob. man.,1} can be equivalently written as
\begin{gather}
d\left(\frac{\d^2\mcF}{\d t^\alpha_{a+1}\d t^\beta_b}\right)=\frac{\d^2\mcF}{\d t^\alpha_a\d t^\mu_0}\eta^{\mu\nu}d\left(\frac{\d^2\mcF}{\d t^\nu_0\d t^\beta_b}\right),\quad 1\le\alpha,\beta\le N,\quad a,b\ge 0,\label{eq:TRR for Frob. man.,2}
\end{gather}
where $d \left( \cdot \right)$ denotes the full differential.

\bigskip

Let $\eps$ be a formal variable and $\tt^\alpha_p:=t^\alpha_p-\delta^{\alpha,1}\delta_{p,1}$. If our Frobenius manifold is {\it conformal}, meaning that the function~$F$ satisfies a certain homogeneity condition, then in~\cite{DZ99} the authors constructed differential operators $L_m$, $m\ge -1$, of the form
\begin{align*}
L_m=&\sum_{p,q\ge 0}\left(\eps^2a_m^{\alpha,p;\beta,q}\frac{\d^2}{\d t^\alpha_p\d t^\beta_q}+b_{m;\alpha,p}^{\beta,q}\tt^\alpha_p\frac{\d}{\d t^\beta_q}+\eps^{-2}c_{m;\alpha,p;\beta,q}\tt^\alpha_p\tt^\beta_q\right)+\\
&+\const\cdot\delta_{m,0},\quad a_m^{\alpha,p;\beta,q},b_{m;\alpha,p}^{\beta,q},c_{m;\alpha,p;\beta,q}\in\mbC,
\end{align*}
satisfying the commutation relations
$$
[L_i,L_j]=(i-j)L_{i+j},\quad i,j\ge -1,
$$
and such that the following equations, called the {\it Virasoro constraints}, are satisfied:
\begin{gather}\label{eq:Virasoro constraints}
\Coef_{\eps^{-2}}\left(\frac{L_m e^{\eps^{-2}\mcF}}{e^{\eps^{-2}\mcF}}\right)=0,\quad m\ge -1.
\end{gather}
We recall the details in Section~\ref{section:Virasoro for Frobenius}.

\bigskip

In Gromov--Witten theory, for each $g\ge 0$ one defines the generating series $\mcF_g(t^*_*)$ of intersection numbers on the moduli space of maps from a Riemann surface of genus $g$ to a target variety, $\mcF_0=\mcF$. The {\it Virasoro conjecture} says that the following equations are satisfied:
\begin{gather}\label{eq:Virasoro in all genera}
L_m e^{\sum_{g\ge 0}\eps^{2g-2}\mcF_g}=0,\quad m\ge -1.
\end{gather}
One can see that equation~\eqref{eq:Virasoro constraints} is the genus $0$ part of equation~\eqref{eq:Virasoro in all genera}. The Virasoro conjecture is proved in a wide class of cases, but is still open in the whole generality.

\bigskip

More recently, a remarkable system of PDEs, extending the WDVV equations~\eqref{eq:WDVV equations}, appeared in the literature. Let $s$ be a formal variable. The {\it open WDVV equations} are the following PDEs for a function $F^o(t^1,\ldots,t^N,s)$:
\begin{align}
\frac{\d^3F}{\d t^\alpha\d t^\beta\d t^\mu}\eta^{\mu\nu}\frac{\d^2F^o}{\d t^\nu\d t^\gamma}+\frac{\d^2F^o}{\d t^\alpha\d t^\beta}\frac{\d^2F^o}{\d s\d t^\gamma}=&\frac{\d^3F}{\d t^\gamma\d t^\beta\d t^\mu}\eta^{\mu\nu}\frac{\d^2F^o}{\d t^\nu\d t^\alpha}+\frac{\d^2F^o}{\d t^\gamma\d t^\beta}\frac{\d^2F^o}{\d s\d t^\alpha},&&1\le\alpha,\beta,\gamma\le N,\label{eq:open WDVV,1}\\
\frac{\d^3F}{\d t^\alpha\d t^\beta\d t^\mu}\eta^{\mu\nu}\frac{\d^2F^o}{\d t^\nu\d s}+\frac{\d^2F^o}{\d t^\alpha\d t^\beta}\frac{\d^2F^o}{\d s^2}=&\frac{\d^2F^o}{\d s\d t^\beta}\frac{\d^2F^o}{\d s\d t^\alpha},&&1\le\alpha,\beta\le N.\label{eq:open WDVV,2}
\end{align}
These equations first appeared in~\cite[Theorem~2.7]{HS12} in the context of open Gromov--Witten theory. The open WDVV equations also appeared in the works~\cite{PST14,BCT17,BCT18}. The solutions of equations~\eqref{eq:open WDVV,1},~\eqref{eq:open WDVV,2}, considered in these works, also satisfy the additional condition
\begin{gather}\label{eq:unit condition for Fo}
\frac{\d^2 F^o}{\d t^1\d t^\alpha}=0,\qquad \frac{\d^2 F^o}{\d t^1\d s}=1.
\end{gather}
\begin{remark}\label{remark about open WDVV}
The works~\cite{PST14,BCT17,BCT18} don't mention the open WDVV explicitly, but, as it is explained in~\cite[Section~4]{Bur18}, the open WDVV equations follow immediately from the open TRR equations~\cite[Theorem~1.5]{PST14}, \cite[Lemma~3.6]{BCT17}, \cite[Theorem~4.1]{BCT18}. 
\end{remark}

\bigskip

There is an open analog of equations~\eqref{eq:TRR for Frob. man.,2}. Let $s_p$, $p\ge 0$, be formal variables, where we identify $s_0=s$. The {\it open topological recursion relations} are the following PDEs for a function~$\mcF^o$, depending on the variables $t^\alpha_p$ and $s_p$, and such that the difference $\left.\mcF^o\right|_{\substack{t^*_{\ge 1}=0\\s_{\ge 1}=0}}-F^o$ is at most linear in the variables $t^1,\ldots,t^N$ and $s$:
\begin{align}
d\left(\frac{\d\mcF^o}{\d t^\alpha_{p+1}}\right)=&\frac{\d^2\mcF}{\d t^\alpha_p\d t^\mu_0}\eta^{\mu\nu}d\left(\frac{\d \mcF^o}{\d t^\nu_0}\right)+\frac{\d\mcF^o}{\d t^\alpha_p}d\left(\frac{\d\mcF^o}{\d s}\right),&1\le\alpha\le N,\quad & p\ge 0,\label{eq:open TRR,1}\\
d\left(\frac{\d\mcF^o}{\d s_{p+1}}\right)=&\frac{\d\mcF^o}{\d s_p}d\left(\frac{\d\mcF^o}{\d s}\right),& & p\ge 0.\label{eq:open TRR,2}
\end{align}
As we already mentioned in Remark~\ref{remark about open WDVV}, these equations appeared in the works~\cite{PST14,BCT17,BCT18}.

\bigskip

The simplest Frobenius manifold has dimension $1$ and the potential $F=F^\pt=\frac{(t^1)^3}{6}$. A natural descendent potential $\mcF^\pt$, associated to it, is given by the generating series of the integrals of monomials in the psi-classes over the moduli space of stable curves of genus $0$. Here "$\pt$" means "point", because such integrals can be considered as the Gromov--Witten invariants of a point. One can easily see that the function $F^o=F^{\pt,o}=t^1 s+\frac{s^3}{6}$ satisfies the open WDVV equations and condition~\eqref{eq:unit condition for Fo}. In~\cite{PST14} the authors, using the intersection theory on the moduli space of stable pointed disks, constructed a solution $\mcF^{\pt,o}$ of the open TRR equations~\eqref{eq:open TRR,1},~\eqref{eq:open TRR,2}. Moreover, they introduced the operators
\begin{gather}\label{eq:open Virasoro operators}
\cL^\pt_m:=L^\pt_m+\eps^{-1}\delta_{m,-1}s+\left(\sum_{i\ge 0}\frac{(i+m+1)!}{i!}s_i\frac{\d}{\d s_{m+i}}+\delta_{m,0}\frac{3}{4}\right)+\eps\frac{3(m+1)!}{4}\frac{\d}{\d s_{m-1}},\quad m\ge -1,
\end{gather}
where $L^\pt_m$ are the Virasoro operators for our Frobenius manifold, and proved the equations~\cite[Theorem~1.1]{PST14}
\begin{gather}\label{eq:open Virasoro in the trivial case}
\Coef_{\eps^{-1}}\left(\frac{\cL^\pt_m e^{\eps^{-2}\mcF^\pt+\eps^{-1}\mcF^{\pt,o}}}{e^{\eps^{-2}\mcF^\pt+\eps^{-1}\mcF^{\pt,o}}}\right)=0,\quad m\ge -1,
\end{gather}
which they called the {\it open Virasoro constraints}.

\begin{remark}
Strictly speaking, in~\cite{PST14} the authors constructed a function $\mcF^{\PST}$, related to our function $\mcF^{\pt,o}$ by $\left.\mcF^{\PST}=\mcF^{\pt,o}\right|_{s_{\ge 1}=0}$. The function $\mcF^{\pt,o}$ can be reconstructed from the function $\mcF^{\PST}$, using the differential equations $\frac{\d\mcF^{\pt,o}}{\d s_n}=\frac{1}{(n+1)!}\left(\frac{\d \mcF^{\pt,o}}{\d s}\right)^{n+1}$. The Virasoro equations, proved in~\cite{PST14}, look as
\begin{align}
\Coef_{\eps^{-1}}\left(\frac{\cL^{\PST}_m e^{\eps^{-2}\mcF^\pt+\eps^{-1}\mcF^{\PST}}}{e^{\eps^{-2}\mcF^\pt+\eps^{-1}\mcF^{\PST}}}\right)=0,\quad \cL^{\PST}_m:=L_m^\pt+\eps^m s\frac{\d^{m+1}}{\d s^{m+1}}+\frac{3m+3}{4}\eps^m\frac{\d^m}{\d s^m},\quad m\ge -1.\label{eq:open Virasoro of PST}
\end{align}
The fact, that equations~\eqref{eq:open Virasoro in the trivial case} and~\eqref{eq:open Virasoro of PST} are equivalent, was noticed in~\cite[Section~5.2]{Bur16}.
\end{remark}

\bigskip

In our paper, we generalize formula~\eqref{eq:open Virasoro in the trivial case} for an arbitrary conformal Frobenius manifold and a solution of the open WDVV equations. We consider an arbitrary conformal Frobenius manifold, an associated descendent potential $\mcF$ and the Virasoro operators $L_m$, $m\ge -1$. Let $F^o$ be a solution of the open WDVV equations, satisfying condition~\eqref{eq:unit condition for Fo} and a certain homogeneity condition, that we will describe later. We will construct a solution $\mcF^o$ of the open TRR equations~\eqref{eq:open TRR,1},~\eqref{eq:open TRR,2} and differential operators~$\cL_m$, $m\ge -1$, of the form
\begin{align}
\cL_m=L_m&+\eps^{-1}\left(\delta_{m,-1}s+\sum_{p\ge 0}d_{m;\alpha,p}\tt^\alpha_p\right)+\label{eq:general open Virasoro}\\
&+\left(\sum_{i\ge 0}\frac{(i+m+1)!}{i!}s_i\frac{\d}{\d s_{m+i}}+\sum_{p,q\ge 0}e_{m;\alpha,p}^q\tt^\alpha_p\frac{\d}{\d s_q}+\delta_{m,0}\frac{3}{4}\right)+\notag\\
&+\eps\left(\sum_{p\ge 0}f_m^{\alpha,p}\frac{\d}{\d t^\alpha_p}+\frac{3(m+1)!}{4}\frac{\d}{\d s_{m-1}}\right)+\notag\\
&+\eps^2\sum_{p,q\ge 0}g_m^{\alpha,p;q}\frac{\d^2}{\d t^\alpha_p\d s_q},\quad d_{m;\alpha,p}, e_{m;\alpha,p}^q, f_m^{\alpha,p}, g_m^{\alpha,p;q}\in\mbC,\notag
\end{align}
such that the equations
\begin{gather*}
\Coef_{\eps^{-1}}\left(\frac{\cL_m e^{\eps^{-2}\mcF+\eps^{-1}\mcF^o}}{e^{\eps^{-2}\mcF+\eps^{-1}\mcF^o}}\right)=0,\quad m\ge -1,
\end{gather*}
hold. The details are given in Section~\ref{section:open Virasoro}, with the main result, formulated in Theorem~\ref{theorem:open Virasoro}.   

\bigskip

It occurs that, given a function $F = F(t^1,\dots,t^N)$, defining a Frobenius manifold and a solution $F^o$ of the open WDVV equations, satisfying~\eqref{eq:unit condition for Fo}, the $(N+1)$-tuple of functions~$\left(\eta^{1\mu}\frac{\d F}{\d t^\mu},\ldots,\eta^{N\mu}\frac{\d F}{\d t^\mu},F^o\right)$ forms a {\it vector potential} of a {\it flat F-manifold}. This was observed by Paolo Rossi. We say that this flat F-manifold \textit{extends} the Frobenius manifold given. In Section~\ref{section:flat F-manifolds} we prove Virasoro type constraints for flat F-manifolds and derive Theorem~\ref{theorem:open Virasoro} as a special case of this result.

\bigskip

Let us return to the particular case, considered in the paper~\cite{PST14}. The construction of the intersection theory on the moduli space of stable pointed disks, given there, can be generalized to higher genera. This has been announced by J.~P. Solomon and R.~J. Tessler, some of the details of their construction are presented in~\cite{Tes15}. As a result, one gets a sequence of functions $\mcF^{\pt,o}_0=\mcF^{\pt,o},\mcF^{\pt,o}_1,\mcF^{\pt,o}_2,\ldots$, and already in~\cite{PST14} the authors conjectured that the following equations should hold:
\begin{gather}\label{eq:open Virasoro for the point in all genera}
\mcL^\pt_m e^{\sum_{g\ge 0}\eps^{2g-2}\mcF^\pt_g+\sum_{g\ge 0}\eps^{g-1}\mcF^{\pt,o}_g}=0,\quad m\ge -1,
\end{gather}
where $\mcF^\pt_g$ is the generating series of the intersection numbers of psi-classes on the moduli space of curves of genus $g$. This conjecture was proved in~\cite{BT17}. 

\bigskip

In Section~\ref{section:open Virasoro in all genera} we discuss a conjectural generalization of equations~\eqref{eq:open Virasoro for the point in all genera} for an arbitrary conformal Frobenius manifold and a solution of the open WDVV equations.

\bigskip

Finally, in Section~\ref{section: examples} we present examples of solutions of the open WDVV equations, for which our main result can be applied. In the separate paper~\cite{BB19} we discuss solutions of the open WDVV equations for the Coxeter groups.

\subsection{Acknowledgements}

We are very grateful to Paolo Rossi for informing us about his observation that solutions of the open WDVV equations correspond to flat F-manifolds of certain type. 

We would like to thank Ezra Getzler for drawing our attention to his paper~\cite{Get04} containing some of the constructions for flat F-manifolds that we present here. He also informed us about Lemma~3.1 from~\cite{Get99} (the proof was provided by E.~Frenkel), which we re-proved in the previous version of the paper. 

We would like to thank Alessandro Arsie, Claus Hertling, Paolo Lorenzoni, Jake Solomon and Ran Tessler for useful discussions. 

This project has received funding from the European Union's Horizon 2020 research and innovation programme under the Marie Sk\l odowska-Curie grant agreement No. 797635. The first named author was partially supported by the Grant RFFI-18-31-20046. The second named author was partially supported by the Grant RFFI-16-01-00409. 


\section{Virasoro constraints for Frobenius manifolds}\label{section:Virasoro for Frobenius}

In this section we review the construction of a descendent potential associated to a solution of the WDVV equations and recall the Virasoro constraints.  

Let us fix $N\ge 1$ and let~$M$ be a simply connected open neighbourhood of a point $(t^1_\orig,\ldots,t^N_\orig)\in\mbC^N$. Denote by $\mcO$ the sheaf of analytic functions on~$\mbC^N$. Consider a solution $F\in\mcO(M)$ of the WDVV equations~\eqref{eq:WDVV equations}, satisfying $\frac{\d^3 F}{\d t^1\d t^\alpha\d t^\beta}=\eta_{\alpha\beta}$. In order to include the case, when~$F$ is a formal power series, in our considerations, we allow~$M$ to be a formal neighbourhood of $(t^1_\orig,\ldots,t^N_\orig)\in\mbC^N$ meaning that $\mcO(M)$ denotes in this case the ring of formal power series in the variables~$(t^\alpha-t^\alpha_{\orig})$.

\subsection{Descendent potential}\label{subsection:descendent potential}

In order to construct a descendent potential, one has to choose an additional structure, called a {\it calibration} of the Frobenius manifold. Denote 
$$
c^\alpha_{\beta\gamma}:=\eta^{\alpha\mu}\frac{\d^3 F}{\d t^\mu\d t^\beta\d t^\gamma}.
$$
A calibration is a collection of functions~$\Omega^{\alpha,0}_{\beta,d}\in\mcO(M)$, $1\le\alpha,\beta\le N$, $d\ge -1$, satisfying the following properties:
\begin{align*}
\Omega^{\alpha,0}_{\beta,-1}=&\delta^\alpha_\beta, && \\
\frac{\d\Omega^{\alpha,0}_{\beta,d}}{\d t^\gamma}=&c^\alpha_{\gamma\mu}\Omega^{\mu,0}_{\beta,d-1}, && d\ge 0,\\
\sum_{\substack{p+q=d\\p,q\ge -1}}(-1)^{p+1}\Omega^{\alpha,0}_{\mu,p}\eta^{\mu\nu}\Omega^{\beta,0}_{\nu,q}=&0, && d\ge -1.
\end{align*}
The space of all calibrations is non-empty and is parameterized by elements $G(z)\in\Mat_{N,N}(\mbC)[[z]]$, satisfying $G(0)=\Id$ and $G(-z)\eta^{-1} G^T(z)=\eta^{-1}$ \cite[Lemma 2.2 and Exercise 2.8]{Dub99}. 

Let us choose a calibration. One can immediately see that $\frac{\d\Omega^{\alpha,0}_{1,0}}{\d t^\beta}=\delta^\alpha_\beta$, which implies that $\Omega^{\alpha,0}_{1,0}-t^\alpha$ is a constant. Let us make the change of coordinates $t^\alpha\mapsto\Omega^{\alpha,0}_{1,0}$, so that we have now $\Omega^{\alpha,0}_{1,0}=t^\alpha$.

Let $v^1,\ldots,v^N$ be formal variables and consider the system of partial differential equations
$$
\frac{\d v^\alpha}{\d t^\beta_d}=\d_x\left(\left.\Omega^{\alpha,0}_{\beta,d}\right|_{t^\gamma=v^\gamma}\right),\quad 1\le\alpha,\beta\le N,\quad d\ge 0,
$$
called the {\it principal hierarchy}. We see that the equations for the flow $\frac{\d}{\d t^1_0}$ look as $\frac{\d v^\alpha}{\d t^1_0}=v^\alpha_x$. This allows to identify $t^1_0=x$. Denote by $(v^\top)^\alpha\in\mcO(M)[[t^*_{\ge 1}]]$ the solution of the principal hierarchy specified by the initial condition
$$
\left.(v^\top)^\alpha\right|_{t^\beta_d=\delta^{\beta,1}\delta_{d,0}x}=\delta^{\alpha,1}x.
$$
We have $\left.(v^\top)^\alpha\right|_{t^*_{\ge 1}=0}=t^\alpha_0$.

Define functions $\Omega_{\alpha,p;\beta,q}\in\mcO(M)$ and $\Omega_{\alpha,p;\beta,q}^\top\in\mcO(M)[[t^*_{\ge 1}]]$,  $p,q\ge 0$, by
\begin{gather*}
\Omega_{\alpha,p;\beta,q}:=\sum_{i=0}^q(-1)^{q-i}\Omega_{\alpha,p+q-i}^{\mu,0}\eta_{\mu\nu}\Omega^{\nu,0}_{\beta,i-1},\qquad\Omega_{\alpha,p;\beta,q}^\top:=\left.\Omega_{\alpha,p;\beta,q}\right|_{t^\gamma=(v^\top)^\gamma}.
\end{gather*}
The descendent potential corresponding to the Frobenius manifold together with the chosen calibration is defined by
$$
\mcF:=\frac{1}{2}\sum_{p,q\ge 0}\tt^\alpha_p\tt^\beta_q\Omega^\top_{\alpha,p;\beta,q}\in\mcO(M)[[t^*_{\ge 1}]],
$$
where we recall that $\tt^\alpha_p=t^\alpha_p-\delta^{\alpha,1}\delta_{p,1}$ \cite[Section 3]{DZ99}. The difference $\left.\mcF\right|_{t^*_{\ge 1}=0}-F$ is at most quadratic in the variables $t^1,\ldots,t^N$ and the function $\mcF$ satisfies equations~\eqref{eq:TRR for Frob. man.,1} together with the equation
\begin{gather*}
\sum_{p\ge 0}\tt^\alpha_{p+1}\frac{\d\mcF}{\d t^\alpha_p}+\frac{1}{2}\eta_{\alpha\beta}t^\alpha t^\beta=0,
\end{gather*}
which is called the {\it string equation}~\cite[Sections 3, 4]{DZ99}. 

\begin{remark}
Strictly speaking, in~\cite{Dub99} and~\cite{DZ99} the authors consider the case of a conformal Frobenius manifold, but one can easily see that the results, discussed in this section, together with their proofs presented in~\cite{Dub99,DZ99}, hold for all not necessarily conformal Frobenius manifolds. We have borrowed the term "calibration" from the paper~\cite{DLYZ16}.
\end{remark}

\subsection{Virasoro constraints}\label{subsection:Virasoro constraints}

The Frobenius manifold is said to be conformal if there exists a vector field $E$ of the form
\begin{gather}\label{eq:Euler vector field}
E=\underbrace{((1-q^\gamma)t^\gamma+r^\gamma)}_{=:E^\gamma}\frac{\d}{\d t^\gamma},\quad q^\gamma,r^\gamma\in\mbC,\quad q^1=0,
\end{gather}
satisfying
\begin{gather*}
L_E F=(3-\delta)F+\frac{1}{2}A_{\alpha\beta}t^\alpha t^\beta+B_\alpha t^\alpha+C,\quad\text{for some $\delta,A_{\alpha\beta},B_\alpha,C\in\mbC$}.
\end{gather*}
The number $\delta$ is often called the {\it conformal dimension}. Denote
$$
\mu^\alpha:=q^\alpha-\frac{\delta}{2},\qquad \mu:=\diag(\mu^1,\ldots,\mu^N). 
$$

In the conformal case there exists a calibration, satisfying the property
\begin{gather}\label{eq:conformal calibration}
E^\theta\frac{\d}{\d t^\theta}\Omega^{\alpha,0}_{\beta,d}=(d+1+\mu^\beta-\mu^\alpha)\Omega^{\alpha,0}_{\beta,d}+\sum_{i=1}^{d+1}\Omega^{\alpha,0}_{\mu,d-i}(R_i)^\mu_\beta,\quad d\ge -1,
\end{gather}
for some matrices $R_n$, $n\ge 1$, satisfying
$$
[\mu,R_n]=nR_n,\qquad \eta R_n \eta^{-1}=(-1)^{n-1}R_n^T.
$$
One can actually see that the matrices $R_n$ are determined uniquely by the functions $\Omega^{\alpha,0}_{\beta,d}$. Note that only finitely many of the matrices $R_n$ are non-zero. The space of all calibrations, satisfying property~\eqref{eq:conformal calibration}, can be explicitly described, see~\cite[Section 2]{Dub99}. A calibration of a conformal Frobenius manifold will always be assumed to satisfy property~\eqref{eq:conformal calibration}.

Let us choose a calibration of our conformal Frobenius manifold and denote 
$$
R:=\sum_{i\ge 1}R_i.
$$
For an arbitrary $N\times N$ matrix $A=(A^\alpha_\beta)$ define matrices $P_m(A,R)$, $m\ge -1$, by the following recursion relation:
\begin{align*}
P_{-1}(A,R)=&\Id,\\
P_{m+1}(A,R)=&R P_m(A,R)+P_m(A,R)\left(A+m+\frac{1}{2}\right),\quad m\ge -1.
\end{align*}
Alternatively, the matrix $P_m(A,R)$ can be defined by 
$$
P_m(A,R)=:\prod_{i=0}^m\left(R+A+i-\frac{1}{2}\right):,
$$
where the symbol $::$ means that, when we take the product, we should place all $R$'s to the left of all $A$'s. Given an integer $p$, define a matrix $[A]_p$ by
$$
([A]_p)^\alpha_\beta:=
\begin{cases}
A^\alpha_\beta,&\text{if $q^\alpha-q^\beta=p$},\\
0,&\text{otherwise}.
\end{cases}
$$
The Virasoro operators $L_m$, $m\ge -1$, of our conformal Frobenius manifold are given by
\begin{align}
L_m:=&\frac{\eps^{-2}}{2}\sum_{d_1,d_2\ge 0}(-1)^{d_1}\left([P_m(\mu+d_2+1,R)]_{m+d_1+d_2+1}\right)^\mu_\beta\eta_{\alpha\mu}\tt^\alpha_{d_1}\tt^\beta_{d_2}+\label{eq:Virasoro1}\\
&+\sum_{\substack{d\ge 0\\0\le k\le m+d}}\left([P_m(\mu+d+1,R)]_{m+d-k}\right)^\beta_\alpha\tt^\alpha_d\frac{\d}{\d t^\beta_k}+\label{eq:Virasoro2}\\
&+\frac{\eps^2}{2}\sum_{\substack{d_1+d_2\le m-1\\d_1,d_2\ge 0}}(-1)^{d_2+1}\left([P_m(\mu-d_2,R)]_{m-1-d_1-d_2}\right)^\alpha_\mu\eta^{\mu\beta}\frac{\d^2}{\d t^\alpha_{d_1}\d t^\beta_{d_2}}+\label{eq:Virasoro3}\\
&+\delta_{m,0}\frac{1}{4}\tr\left(\frac{1}{4}-\mu^2\right),\quad m\ge -1,\notag
\end{align}
and equations~\eqref{eq:Virasoro constraints} hold~\cite[page 428]{DZ99}.

\begin{remark}
One can easily see that the last term $\delta_{m,0}\frac{1}{4}\tr\left(\frac{1}{4}-\mu^2\right)$ in the expression for the Virasoro operators $L_m$ doesn't play a role in equations~\eqref{eq:Virasoro constraints}. However, this term plays a role in the Virasoro constraints in all genera~\eqref{eq:Virasoro in all genera}. We discuss it in more details in Section~\ref{section:open Virasoro in all genera}. 
\end{remark}


\section{Open Virasoro constraints}\label{section:open Virasoro}

Here we present our construction of an {\it open descendent potential} associated to a solution of the open WDVV equations and present our main result -- the open Virasoro constraints, given in Theorem~\ref{theorem:open Virasoro}. 

Consider a conformal Frobenius manifold~$M$ together with its calibration from Section~\ref{subsection:Virasoro constraints}. Let $U$ be a simply connected open neighbourhood of a point $s_{\orig}\in\mbC$. Suppose that a function $F^o(t^1,\ldots,t^N,s)\in\mcO(M\times U)$ satisfies the open WDVV equations~\eqref{eq:open WDVV,1},~\eqref{eq:open WDVV,2} and condition~\eqref{eq:unit condition for Fo}. Let us assume that the function $F^o$ satisfies the homogeneity condition
\begin{gather}\label{eq:homogeneity for Fo}
E^\gamma\frac{\d F^o}{\d t^\gamma}+\underbrace{\left(\frac{1-\delta}{2}s+r^{N+1}\right)}_{=:E^{N+1}}\frac{\d F^o}{\d s}=\frac{3-\delta}{2}F^o+D_\alpha t^\alpha+\tD s+E,\quad\text{for some $r^{N+1},D_\alpha,\tD,E\in\mbC$}.
\end{gather}
This condition holds in the examples, considered in the papers~\cite{HS12,PST14,BCT17,BCT18}. Actually, in these examples the constant $r^{N+1}$ is zero, but this is not needed in our considerations.

In order to construct an open descendent potential $\mcF^o$ we need an additional structure, similar to a calibration of a Frobenius manifold. It will be convenient for us to denote the variable $s_d$ by~$t^{N+1}_d$. We adopt the conventions 
\begin{gather*}
\Omega^{\alpha,0}_{N+1,d}:=0,\quad 1\le\alpha\le N,\quad d\ge -1,\qquad\qquad \mu^{N+1}:=\frac{1}{2},
\end{gather*}
and define a diagonal $(N+1)\times(N+1)$ matrix $\tmu$ by
$$
\tmu:=\diag(\mu^1,\ldots,\mu^{N+1}).
$$

\begin{definition}
A calibration of the function $F^o$ is a sequence of functions $\Phi_{\beta,d}\in\mcO(M\times U)$, $1\le\beta\le N+1$, $d\ge -1$, satisfying the properties
\begin{align}
\Phi_{\beta,-1}=&\delta_{\beta,N+1}, && \label{eq:open calibration,property 1}\\
\frac{\d\Phi_{\beta,d}}{\d t^\gamma}=&\sum_{1\le\mu\le N}\frac{\d^2 F^o}{\d t^\gamma\d t^\mu}\Omega^{\mu,0}_{\beta,d-1}+\frac{\d^2 F^o}{\d t^\gamma\d t^{N+1}}\Phi_{\beta,d-1}, && d\ge 0,\label{eq:open calibration,property 2}\\
E^\mu\frac{\d\Phi_{\beta,d}}{\d t^\mu}=&\left(d+\frac{1}{2}+\mu^\beta\right)\Phi_{\beta,d}+\sum_{i=1}^{d+1}\Phi_{\mu,d-i}(\tR_i)^\mu_\beta, && d\ge -1,\label{eq:open calibration, property 3}
\end{align}
for some $(N+1)\times(N+1)$ matrices $\tR_n$, $n\ge 1$, satisfying 
$$
(\tR_n)^\alpha_\beta=
\begin{cases}
(R_n)^\alpha_\beta,&\text{if $1\le\alpha,\beta\le N$},\\
0,&\text{if $\beta=N+1$},
\end{cases}
\qquad [\tmu,\tR_n]=n\tR_n.
$$
\end{definition}
One can easily see that, for a calibration of the function $F^o$, matrices $\tR_n$ are uniquely determined by the functions $\Phi_{\beta,d}$.
\begin{lemma}\label{lemma:existence of open calibrations}
The space of calibrations of the function $F^o$ is non-empty.
\end{lemma}
The proof of the lemma will be given in Section~\ref{subsubsection:calibrations of extensions}.

Consider a calibration of the function $F^o$. One can easily see that $\frac{\d\Phi_{1,0}}{\d t^\gamma}=\delta_{\gamma,N+1}$, which implies that $\Phi_{1,0}-t^{N+1}$ is a constant. Let us make the change of coordinates $t^{N+1}\mapsto\Phi_{1,0}$, so that we have now $\Phi_{1,0}=t^{N+1}$.

Let $\tv^1,\tv^2,\ldots,\tv^{N+1}$ be formal variables and consider the system of partial differential equations
$$
\left\{\begin{aligned}
\frac{\d\tv^\alpha}{\d t^\beta_d}=&\d_x\left(\left.\Omega^{\alpha,0}_{\beta,d}\right|_{t^\gamma=\tv^\gamma}\right),& 1\le\alpha\le N,\quad & 1\le\beta\le N+1,\quad d\ge 0,\\
\frac{\d\tv^{N+1}}{\d t^\beta_d}=&\d_x\left(\left.\Phi_{\beta,d}\right|_{t^\gamma=\tv^\gamma}\right),& & 1\le\beta\le N+1,\quad d\ge 0.
\end{aligned}\right.
$$
Let $(\tv^\top)^\alpha$ be the solution, specified by the initial condition
$$
\left.(\tv^\top)^\alpha\right|_{t^\beta_d=\delta^{\beta,1}\delta_{d,0}x}=\delta^{\alpha,1}x.
$$
Clearly, $(\tv^\top)^\alpha=(v^\top)^\alpha$, for $1\le\alpha\le N$. Define
$$
\Phi_{\beta,d}^\top:=\left.\Phi_{\beta,d}\right|_{t^\gamma=(\tv^\top)^\gamma}.
$$
Then we define the open descendent potential $\mcF^o$ by
$$
\mcF^o:=\sum_{d\ge 0}\tt^\alpha_d\Phi^\top_{\alpha,d}.
$$
\begin{lemma}\label{lemma:open descendent potential}
The function $\mcF^o$ satisfies equations~\eqref{eq:open TRR,1},~\eqref{eq:open TRR,2} and the difference $\left.\mcF^o\right|_{\substack{t^*_{\ge 1}=0\\s_{\ge 1}=0}}-F^o$ is at most linear in the variables $t^1,\ldots,t^N$ and $s$.
\end{lemma}
We will prove the lemma in Section~\ref{subsection:descendent vector potential}.

Let $L_m$, $m\ge -1$, be the Virasoro operators for our conformal Frobenius manifold and define operators $\cL_m$, $m\ge -1$, by
\begin{align}
\cL_m:=L_m&+\eps^{-1}\sum_{d\ge 0}\sum_{1\le\alpha\le N+1}\left(\left[P_m(\tmu+d+1,\tR)\right]_{m+d+1}\right)^{N+1}_\alpha\tt^\alpha_d\label{eq:exact open Virasoro}\\
&+\left(\sum_{\substack{d\ge 0\\0\le k\le m+d}}\sum_{1\le\alpha\le N+1}\left(\left[P_m(\tmu+d+1,\tR)\right]_{m+d-k}\right)^{N+1}_\alpha\tt^\alpha_d\frac{\d}{\d t^{N+1}_k}+\delta_{m,0}\frac{3}{4}\right)\notag\\
&+\eps\left(\sum_{0\le k\le m}\sum_{1\le\alpha,\mu\le N}(-1)^{k+1}\left(\left[P_m(\tmu-k,\tR)\right]_{m-k}\right)^{N+1}_\alpha\eta^{\alpha\mu}\frac{\d}{\d t^\mu_k}+\frac{3(m+1)!}{4}\frac{\d}{\d t^{N+1}_{m-1}}\right)\notag\\
&+\eps^2\sum_{\substack{d_1,d_2\ge 0\\d_1+d_2\le m-1}}\sum_{1\le\alpha,\mu\le N}(-1)^{d_2+1}\left(\left[P_m(\tmu-d_2,\tR)\right]_{m-1-d_1-d_2}\right)^{N+1}_\alpha\eta^{\alpha\mu}\frac{\d^2}{\d t^{N+1}_{d_1}\d t^\mu_{d_2}}.\notag
\end{align}
\begin{theorem}\label{theorem:open Virasoro}
We have
\begin{gather}\label{eq:open Virasoro}
\Coef_{\eps^{-1}}\left(\frac{\mcL_m e^{\eps^{-2}\mcF+\eps^{-1}\mcF^o}}{e^{\eps^{-2}\mcF+\eps^{-1}\mcF^o}}\right)=0,\quad m\ge -1.
\end{gather}
\end{theorem}

\begin{remark}
Since $\tR^\alpha_{N+1}=0$, we have
\begin{align*}
&\left(\left[P_m(\tmu+d+1,\tR)\right]_{m+d+1}\right)^{N+1}_{N+1}=\delta_{m,-1}\delta_{d,0},&& m\ge -1,\,\, d\ge 0,\\
&\left(\left[P_m(\tmu+d+1,\tR)\right]_{m+d-k}\right)^{N+1}_{N+1}=\delta_{k,m+d}\frac{(d+m+1)!}{d!},&& m\ge -1,\,\, d\ge 0, \,\, 0\le k\le m+d.
\end{align*}
Therefore, the operators $\mcL_m$, given by~\eqref{eq:exact open Virasoro}, have the form~\eqref{eq:general open Virasoro}.
\end{remark}

\begin{remark}
One can easily see that the expression $\frac{\mcL_m e^{\eps^{-2}\mcF+\eps^{-1}\mcF^o}}{e^{\eps^{-2}\mcF+\eps^{-1}\mcF^o}}$ has the form
$$
\frac{\mcL_m e^{\eps^{-2}\mcF+\eps^{-1}\mcF^o}}{e^{\eps^{-2}\mcF+\eps^{-1}\mcF^o}}=\sum_{i\ge -2}\eps^i f_i(t^*_*),
$$
for some functions $f_i$ depending on the variables $t^\alpha_p$, $1\le\alpha\le N+1$, $p\ge 0$. Note that the vanishing of the function $f_{-2}$ is equivalent to the Virasoro equations for the descendent potential~$\mcF$, 
$$
f_{-2}=\Coef_{\eps^{-2}}\left(\frac{L_m e^{\eps^{-2}\mcF}}{e^{\eps^{-2}\mcF}}\right)=0.
$$
\end{remark}

\begin{remark}
One can see that the terms $\delta_{m,0}\frac{3}{4}$ and $\eps\frac{3(m+1)!}{4}\frac{\d}{\d t^{N+1}_{m-1}}$ in the expression for the operator $\mcL_m$ don't play a role in equations~\eqref{eq:open Virasoro}. However, they play a role in our conjectural open Virasoro constraints in all genera, which we discuss in Section~\ref{section:open Virasoro in all genera}.
\end{remark}

In the next section we derive Virasoro type equations for flat F-manifolds and then get Theorem~\ref{theorem:open Virasoro} as a special case of this result.


\section{Virasoro type constraints for flat F-manifolds}\label{section:flat F-manifolds}

In this section we recall the definition of a flat F-manifold and show how such an object can be associated to a solution of the open WDVV equations. We then present a construction of {\it descendent vector potentials}, corresponding to a flat F-manifold, and prove Virasoro type constraints for them. The open Virasoro constraints from Theorem~\ref{theorem:open Virasoro} are derived as a corollary of this result. 

\subsection{Flat F-manifolds}

Here we recall the definition and the main properties of flat F-manifolds. We refer a reader to the papers~\cite{Man05,AL18} for more details. Flat F-manifolds are also studied in the paper~\cite{Get04}, where they are called Dubrovin manifolds.

\begin{definition}
A flat F-manifold $(M,\nabla,\circ)$ is the datum of an analytic manifold $M$, an analytic connection $\nabla$ in the tangent bundle $T M$, an algebra structure $(T_p M,\circ)$ with unit $e$ on each tangent space, analytically depending on the point $p\in M$, such that the one-parameter family of connections $\nabla+z\circ$ is flat and torsionless for any $z\in\mbC$, and $\nabla e=0$.
\end{definition}

From the flatness and the torsionlessness of $\nabla+z\circ$ one can deduce the commutativity and the associativity of the algebras $(T_pM,\circ)$. Moreover, if one choses flat coordinates $t^\alpha$, $1\le\alpha\le N$, $N=\dim M$, for the connection $\nabla$, with $e = \frac{\d}{\d t^1}$, then it is easy to see that locally there exist analytic functions $F^\alpha(t^1,\ldots,t^N)$, $1\leq\alpha\leq N$, such that the second derivatives 
\begin{gather}\label{eq:structure constants of flat F-man}
c^\alpha_{\beta\gamma}=\frac{\d^2 F^\alpha}{\d t^\beta \d t^\gamma}
\end{gather}
give the structure constants of the algebras $(T_p M,\circ)$,
\begin{gather*}
\frac{\d}{\d t^\beta}\circ\frac{\d}{\d t^\gamma}=c^\alpha_{\beta\gamma}\frac{\d}{\d t^\alpha}.
\end{gather*}
From the associativity of the algebras $(T_p M,\circ)$ and the fact that the vector $\frac{\d}{\d t^1}$ is the unit it follows that
\begin{align}
\frac{\d^2 F^\alpha}{\d t^1\d t^\beta} &= \delta^\alpha_\beta, && 1\leq \alpha,\beta\leq N,\label{eq:axiom1 of flat F-man}\\
\frac{\d^2 F^\alpha}{\d t^\beta \d t^\mu} \frac{\d^2 F^\mu}{\d t^\gamma \d t^\delta} &= \frac{\d^2 F^\alpha}{\d t^\gamma \d t^\mu} \frac{\d^2 F^\mu}{\d t^\beta \d t^\delta}, && 1\leq \alpha,\beta,\gamma,\delta\leq N.\label{eq:axiom2 of flat F-man}
\end{align}
The $N$-tuple of functions $(F^1,\ldots,F^N)$ is called the vector potential of the flat F-manifold.

Conversely, if $M$ is an open subset of $\mbC^N$ and $F^1,\ldots,F^N\in\mcO(M)$ are functions, satisfying equations~\eqref{eq:axiom1 of flat F-man} and~\eqref{eq:axiom2 of flat F-man}, then these functions define a flat F-manifold $(M,\nabla,\circ)$ with the connection~$\nabla$, given by $\nabla_{\frac{\d}{\d t^\alpha}}\frac{\d}{\d t^\beta}=0$, and the multiplication $\circ$, given by the structure constants~\eqref{eq:structure constants of flat F-man}.

A flat F-manifold, given by a vector potential $(F^1,\ldots,F^N)$, is called {\it conformal}, if there exists a vector field of the form~\eqref{eq:Euler vector field} such that
\begin{gather}\label{eq:homogeneity for F-man}
E^\mu\frac{\d F^\alpha}{\d t^\mu}=(2-q^\alpha)F^\alpha+A^\alpha_\beta t^\beta+B^\alpha,\quad\text{for some $A^\alpha_\beta,B^\alpha\in\mbC$}.
\end{gather}
\begin{remark}
A Frobenius manifold with a potential $F$ and a metric $\eta$ defines the flat F-manifold with the vector potential $F^\alpha=\eta^{\alpha\mu}\frac{\d F}{\d t^\mu}$. If the Frobenius manifold is conformal, then the associated flat F-manifold is also conformal. This follows from the property
$$
(q^\alpha+q^\beta-\delta)\eta_{\alpha\beta}=0.
$$
\end{remark}

A point $p\in M$ of an $N$-dimensional flat F-manifold $(M,\nabla,\circ)$ is called \textit{semisimple} if $T_pM$ has a basis of idempotents $\pi_1,\dots,\pi_N$, satisfying $\pi_k \circ \pi_l = \delta_{k,l} \pi_k$. Moreover, locally around such a point one can choose coordinates $u^i$ such that $\frac{\d}{\d u^k}\circ\frac{\d}{\d u^l}=\delta_{k,l}\frac{\d}{\d u^k}$. These coordinates are called the {\it canonical coordinates}. In particular, this means that the semisimplicity is an open property of a point. The flat F-manifold $M$ is called semisimple, if a generic point of $M$ is semisimple.

\begin{remark}
In the semisimple case, a conformal flat F-manifold is a special case of a bi-flat F-manifold, see~\cite[Theorem 4.4]{AL17}.
\end{remark}

\subsection{Extensions of flat F-manifolds and the open WDVV equations}

Consider a flat F-manifold structure, given by a vector potential $(F^1,\ldots,F^{N+1})$ on an open subset ${M\times U\in\mbC^{N+1}}$, where $M$ and $U$ are open subsets of $\mbC^N$ and $\mbC$, respectively. Suppose that the functions $F^1,\ldots,F^N$ don't depend on the variable $t^{N+1}$, varying in $U$. Then the functions $F^1,\ldots,F^N$ satisfy equations~\eqref{eq:axiom2 of flat F-man} and, thus, define a flat F-manifold structure on~$M$. In this case we call the flat F-manifold structure on $M\times U$ an {\it extension} of a flat F-manifold structure on $M$.

Consider the flat F-manifold, associated to a Frobenius manifold, given by a potential $F(t^1,\ldots,t^N)\in\mcO(M)$ and a metric $\eta$, $F^\alpha=\eta^{\alpha\mu}\frac{\d F}{\d t^\mu}$, $1\le\alpha\le N$. It is easy to check that a function $F^o(t^1,\ldots,t^N,s)\in\mcO(M\times U)$ satisfies equations~\eqref{eq:open WDVV,1},~\eqref{eq:open WDVV,2} and~\eqref{eq:unit condition for Fo} if and only if the $(N+1)$-tuple $(F^1,\ldots,F^N,F^o)$ is a vector potential of a flat F-manifold. Recall that here we identify $s=t^{N+1}$. This defines a correspondence between solutions of the open WDVV equations, satisfying property~\eqref{eq:unit condition for Fo}, and flat F-manifolds, extending the Frobenius manifold given. This observation belongs to Paolo Rossi. 

\subsection{Calibration of a flat F-manifold}

In this section we introduce the notion of a {\it calibration} of a flat F-manifold. As far as we know, calibrations of flat F-manifolds were first considered in~\cite{Get04}, where they are called fundamental solutions.

\subsubsection{General case}

Let $M$ be a simply connected open neighbourhood of a point $(t^1_\orig,\ldots,t^N_\orig)\in\mbC^N$ and consider a flat F-manifold structure on $M$ given by a vector potential $(F^1,\ldots,F^N)$, $F^\alpha\in\mcO(M)$. A calibration of our flat F-manifold is a collection of functions~$\Omega^{\alpha,d}_{\beta,0}\in\mcO(M)$, $1\le\alpha,\beta\le N$, $d\ge -1$, satisfying $\Omega^{\alpha,-1}_{\beta,0}=\delta^\alpha_\beta$ and the property
\begin{gather}\label{eq:calibration of flat F-man,eq2}
\frac{\d\Omega^{\alpha,d}_{\beta,0}}{\d t^\gamma}=c^\mu_{\gamma\beta}\Omega^{\alpha,d-1}_{\mu,0}, \quad d\ge 0.
\end{gather}

Let us describe the space of all calibrations of our flat F-manifold. Denote by $\onabla$ the family of connections, depending on a formal parameter $z$, given by
$$
\onabla_X Y:=\nabla_X Y+z X\circ Y,
$$
where $X$ and $Y$ are vector fields on $M$. Then equation~\eqref{eq:calibration of flat F-man,eq2} is equivalent to the flatness, with respect to~$\onabla$, of the $1$-forms
$$
\left(\sum_{d\ge 0}\Omega^{\alpha,d-1}_{\beta,0}z^d\right)d t^\beta,\quad 1\le\alpha\le N.
$$
From the flatness of the connection $\onabla$ it follows that a calibration of our flat F-manifold exists. In order to describe the whole space of calibrations, introduce $N\times N$ matrices $\Omega^d_0$, $d\ge -1$, by $(\Omega^d_0)^\alpha_\beta:=\Omega^{\alpha,d}_{\beta,0}$. Then equation~\eqref{eq:calibration of flat F-man,eq2} can be written as
\begin{gather}\label{eq:TRR for Omega^d_0}
\frac{\d}{\d t^\gamma}\left(\sum_{d\ge 0}\Omega^{d-1}_0z^d\right)=z\left(\sum_{d\ge 0}\Omega^{d-1}_0z^d\right) C_\gamma,
\end{gather}
where $C_\gamma:=(c_{\gamma\beta}^\alpha)$. From this it becomes clear that the generating series $\sum_{d\ge 0}\Omega^{d-1}_0z^d$ is determined uniquely up to a transformation of the form
$$
\sum_{d\ge 0}\Omega^{d-1}_0z^d\mapsto G(z)\left(\sum_{d\ge 0}\Omega^{d-1}_0z^d\right),\quad G(z)\in\Mat_{N,N}(\mbC)[[z]],\quad G(0)=\Id.
$$

Let us introduce matrices $\Omega^0_d=(\Omega^{\alpha,0}_{\beta,d})$, $d\ge 0$, by the equation
\begin{gather}\label{eq:upper-lower relation}
\left(\Id+\sum_{d\ge 1}(-1)^d\Omega^0_{d-1}z^d\right)\left(\Id+\sum_{d\ge 1}\Omega^{d-1}_0 z^d\right)=\Id.
\end{gather}
By definition, we put $\Omega^0_{-1}:=\Id$. From equation~\eqref{eq:TRR for Omega^d_0} we get
\begin{gather}\label{eq:TRR for Omega^0_d}
\frac{\d\Omega^0_d}{\d t^\gamma}=C_\gamma\Omega^0_{d-1}, \quad d\ge 0.
\end{gather}

\subsubsection{Conformal case}

Suppose now that our flat F-manifold is conformal. Introduce a diagonal matrix $Q$ by
$$
Q:=\diag(q^1,\ldots,q^N).
$$
\begin{proposition}\cite{Get04}\label{proposition:calibration of conformal F-man}
There exists a calibration such that 
\begin{gather}\label{eq:homogeneity for Omega^d_0}
E^\theta\frac{\d\Omega^d_0}{\d t^\theta}=(d+1)\Omega^d_0+[\Omega^d_0,Q]+\sum_{i=1}^{d+1}(-1)^{i-1}R_i\Omega^{d-i}_0,\quad d\ge -1,
\end{gather}
for some matrices $R_n$, $n\ge 1$, satisfying $[Q,R_n]=nR_n$.
\end{proposition}
\begin{proof}
Similarly to the work~\cite[pages 310, 312]{Dub99}, the proposition is proved by considering a certain flat connection on $M\times\mbC^*$. 

Introduce a family of connections $\tnabla^\lambda$, depending on a complex parameter~$\lambda$, on~$M\times\mbC^*$ by
\begin{gather*}
\tnabla^\lambda_X Y:=\nabla_X Y+z X\circ Y,\qquad\tnabla^\lambda_{\frac{\d}{\d z}}Y:=\frac{\d Y}{\d z}+E\circ Y+\frac{\lambda-Q}{z}Y,\qquad\tnabla^\lambda_X\frac{\d}{\d z}=\tnabla^\lambda_{\frac{\d}{\d z}}\frac{\d}{\d z}:=0,
\end{gather*}
where $X$ and $Y$ are vector fields on $M\times\mbC^*$ having zero component along $\mbC^*$.
\begin{remark}
Note that for the flat F-manifold, associated to a conformal Frobenius manifold~$M$, the connection $\tnabla^{\frac{\delta}{2}}$ coincides with the flat connection $\tnabla$ on $M\times\mbC^*$ from the paper~\cite[page 310]{Dub99}).
\end{remark}
\begin{lemma}
The connection $\tnabla^\lambda$ is flat.
\end{lemma}
\begin{proof}
Direct computation analogous to the one from~\cite[proof of Proposition~2.1]{Dub99}.
\end{proof}
A differential form $\zeta_\alpha(t^*,z)dt^\alpha$ on $M\times\mbC^*$ is flat with respect to the connection $\tnabla^\lambda$ if and only if the following equations are satisfied:
\begin{align*}
\frac{\d\zeta_\alpha}{\d t^\beta}=&z c^\gamma_{\beta\alpha}\zeta_\gamma,\notag\\
\frac{\d\zeta_\alpha}{\d z}=&\mcU^\gamma_\alpha\zeta_\gamma+\frac{\lambda-q^\alpha}{z}\zeta_\alpha,
\end{align*}
where $\mcU^\alpha_\beta:=E^\mu c^\alpha_{\mu\beta}$. Denote by $\zeta$ the row vector $(\zeta_1,\ldots,\zeta_N)$, then the last two equations can be written as
\begin{align}
\frac{\d\zeta}{\d t^\beta}=&z \zeta C_\beta,\label{eq:equation for zeta,1}\\
\frac{\d\zeta}{\d z}=&\zeta\left(\mcU+\frac{\lambda-Q}{z}\right),\label{eq:equation for zeta,2}
\end{align}
where $\mcU:=(\mcU^\alpha_\beta)$.

Let us construct a certain matrix solution of the system~\eqref{eq:equation for zeta,1},~\eqref{eq:equation for zeta,2}. We first consider equation~\eqref{eq:equation for zeta,2} along the punctured line $(t^1_\orig,\ldots,t^N_\orig)\times\mbC^*\subset M\times\mbC^*$. Let $U_1:=\mcU|_{t^\alpha=t^\alpha_{\orig}}$ and consider the equation
\begin{gather}\label{eq:equation for xi}
\frac{\d\xi}{\d z}=\xi\left(U_1+\frac{\lambda-Q}{z}\right)\quad\Leftrightarrow\quad \frac{\d\xi^T}{\d z}=\left(U_1^T+\frac{\lambda-Q}{z}\right)\xi^T,
\end{gather}
for an $N\times N$ matrix $\xi$. Then there exists a transformation $\xi'=G(z)\xi^T$ with $G(z)\in\Mat_{N,N}(\mbC)[[z]]$, $G(0)=\Id$, that transforms equation~\eqref{eq:equation for xi} to
\begin{gather}\label{eq:equation for txi}
\frac{\d\xi'}{\d z}=\left(\frac{\lambda-Q}{z}+\sum_{n\ge 1}(-1)^{n-1}R_n^T z^{n-1}\right)\xi',
\end{gather}
where matrices $R_n$ satisfy $[Q,R_n]=nR_n$ (see e.g.~\cite[Lemma~2.5]{Dub99}). The fact that $G(z)$ and the matrices $R_n$ can be chosen not to depend on $\lambda$ is obvious. The matrix $\xi'=z^{\lambda-Q}z^{\oR^T}$, where $\oR:=\sum_{n\ge 1}(-1)^{n-1}R_n$, satisfies equation~\eqref{eq:equation for txi}. Therefore, the matrix $\xi=z^\oR z^{\lambda-Q}(G^T(z))^{-1}$ satisfies equation~\eqref{eq:equation for xi}. 

Using equation~\eqref{eq:equation for zeta,1}, we can extend the constructed function $\xi$ on the punctured line $(t^1_\orig,\ldots,t^N_\orig)\times\mbC^*\subset M\times\mbC^*$ to a function~$\zeta$ on the whole space $M\times\mbC^*$. The function~$\zeta$ has the form 
\begin{gather}\label{eq:zeta solution}
\zeta=z^\oR z^{\lambda-Q}\left(\sum_{d\ge 0}\Omega^{d-1}_0z^d\right),
\end{gather}
where $\Omega^d_0\in\Mat_{N,N}(\mcO(M))$, $\Omega^{-1}_0=\Id$ and $\left.\left(\sum_{d\ge 0}\Omega^{d-1}_0z^d\right)\right|_{t^\alpha=t^\alpha_{\orig}}=(G^T(z))^{-1}$. Since the connection~$\tnabla^\lambda$ is flat, the function $\zeta$ satisfies equation~\eqref{eq:equation for zeta,2}.

Equation~\eqref{eq:equation for zeta,1} for a function $\zeta$ of the form~\eqref{eq:zeta solution} implies that the sequence of matrices $\Omega^d_0$, $d\ge -1$, is a calibration of our flat F-manifold. Equation~\eqref{eq:equation for zeta,2} gives
$$
\left(\frac{\lambda-Q}{z}+\sum_{n\ge 1}(-1)^{n-1}R_nz^{n-1}\right)\left(\sum_{d\ge 0}\Omega^{d-1}_0z^d\right)+\sum_{d\ge 0}(d+1)\Omega^d_0z^d=\left(\sum_{d\ge 0}\Omega^{d-1}_0z^d\right)\left(\mcU+\frac{\lambda-Q}{z}\right).
$$
Taking the coefficient of $z^d$, $d\ge 0$, in this equation, we get
$$
\Omega^{d-1}_0\mcU=(d+1)\Omega^d_0+[\Omega^d_0,Q]+\sum_{i=1}^{d+1}(-1)^{i-1}R_i\Omega^{d-i}_0,\quad d\ge 0.
$$
By~\eqref{eq:equation for zeta,1}, the left-hand side is equal to $E^\theta\frac{\d\Omega^d_0}{\d t^\theta}$. This completes the proof of the proposition.
\end{proof} 

A calibration of a conformal flat F-manifold will always be assumed to satisfy property~\eqref{eq:homogeneity for Omega^d_0}.

From equation~\eqref{eq:homogeneity for Omega^d_0} it is easy to deduce that
\begin{gather}\label{eq:homogeneity for Omega^0_d}
E^\theta\frac{\d\Omega^0_d}{\d t^\theta}=(d+1)\Omega^0_d+[\Omega^0_d,Q]+\sum_{i=1}^{d+1}\Omega^0_{d-i} R_i, \quad d\ge -1.
\end{gather}

\begin{remark}
We see that a calibration of a conformal Frobenius manifold is the same as a calibration of the associated flat F-manifold, satisfying the additional properties $\eta\Omega^0_d\eta^{-1}=(\Omega^d_0)^T$ and $\eta R_n\eta^{-1}=(-1)^{n-1}R_n^T$.
\end{remark}

\subsubsection{Calibrations of extensions of flat F-manifolds}\label{subsubsection:calibrations of extensions}

For an $(N+1)\times(N+1)$ matrix $A$ denote by $\pi_N(A)$ the $N\times N$ matrix formed by the first $N$ raws and the first $N$ columns of $A$.

\begin{lemma}\label{lemma:calibration of extension}
Consider a flat F-manifold, given by a vector potential $(F^1,\ldots,F^N)$, $F^\alpha\in\mcO(M)$, and its extension with a vector potential $(F^1,\ldots,F^{N+1})$, $F^{N+1}\in\mcO(M\times U)$, where the open subsets $M\in\mbC^N$ and $U\in\mbC$ are simply connected. Suppose that the flat F-manifold $M\times U$ is conformal with an Euler vector field $E=\sum_{\alpha=1}^{N+1}((1-q^\alpha)t^\alpha+r^\alpha)\frac{\d}{\d t^\alpha}$. Let us also fix a calibration of the flat F-manifold $M$, given by matrices $\Omega^d_0$ and $R_n$. 

Then there exists a calibration of the flat F-manifold $M\times U$ with matrices $\tOmega^d_0$ and $\tR_n$ satisfying the properties 
\begin{gather}\label{eq:property of calibration of extension}
\pi_N(\tOmega^d_0)=\Omega^d_0,\qquad \pi_N(\tR_n)=R_n,\qquad \tOmega^{\le N,d}_{N+1,0}=(\tR_n)^{\le N+1}_{N+1}=0.
\end{gather}
\end{lemma}
\begin{proof}
Consider the construction of a calibration of the conformal flat F-manifold $M\times U$ from the proof of Proposition~\ref{proposition:calibration of conformal F-man} in more details. So we consider the differential equation
$$
\frac{\d\txi^T}{\d z}=\left(\tU_1^T+\frac{\lambda-\tQ}{z}\right)\txi^T
$$
for an $(N+1)\times(N+1)$ matrix $\txi$, where $(\tU_1)^\alpha_\beta:=(E^\mu c^\alpha_{\mu\beta})|_{t^\theta=t^\theta_{\orig}}$, $1\le\alpha,\beta\le N+1$, and $\tQ:=\diag(q^1,\ldots,q^{N+1})$. A transformation $\xi'=\tG(z)\txi^T$, $\tG(z)=\Id+\sum_{n\ge 1}\tG_n z^n$, transforming this differential equation to the form
\begin{gather*}
\frac{\d\xi'}{\d z}=\left(\frac{\lambda-\tQ}{z}+\sum_{n\ge 1}(-1)^{n-1}\tR_n^T z^{n-1}\right)\xi',
\end{gather*}
where $[\tQ,\tR_n]=n\tR_n$, is determined by the recursion relation~\cite[equation~(2.53)]{Dub99}
\begin{gather}\label{eq:recursion for extension}
(-1)^{n-1}\tR_n^T=\delta_{n,1}\tU_1^T+n\tG_n+[\tQ,\tG_n]+\sum_{k=1}^{n-1}\left(\tG_{n-k}\delta_{k,1}\tU_1^T-(-1)^{k-1}\tR_k^T\tG_{n-k}\right),\quad n\ge 1.
\end{gather}
If one has computed the matrices $\tG_i$ and $\tR_i$ for $i<n$, then equation~\eqref{eq:recursion for extension} determines the matrix $\tR_n$ and the elements $(\tG_n)^\alpha_\beta$ with $q^\alpha-q^\beta\ne -n$. The elements $(\tG_n)^\alpha_\beta$ with $q^\alpha-q^\beta=-n$ can be chosen arbtitrarily. Note that $(\tU_1)^{\le N}_{N+1}=0$. Therefore, if $(\tR_i)^{\le N}_{N+1}=(\tG_i)^{N+1}_{\le N}=0$ for $i<n$, then $(\tR_n)^{\le N}_{N+1}=0$ and, choosing the elements $(\tG_n)^{N+1}_\alpha$ with $q^\alpha-q^\beta=-n$ to be zero, we can guarantee that $(\tG_n)^{N+1}_{\le N}=0$. 

Let $U_1:=\pi_N(\tU_1)$ and $Q:=\pi_N(\tQ)$. We know that the $N\times N$ matrices $G_n$, given by $\left.\left(\sum_{d\ge 0}\Omega^{d-1}_0z^d\right)\right|_{t^\alpha=t^\alpha_{\orig}}=(\Id+\sum_{n\ge 1}G_n^T z^n)^{-1}$, together with the matrices $R_n$ satisfy the equations
\begin{gather*}
(-1)^{n-1}R_n^T=\delta_{n,1}U_1^T+n G_n+[Q,G_n]+\sum_{k=1}^{n-1}\left(G_{n-k}\delta_{k,1}U_1^T-(-1)^{k-1}R_k^T G_{n-k}\right),\quad n\ge 1.
\end{gather*}
Therefore, there exist matrices $\tR_n$ and $\tG_n$, $n\ge 1$, satisfying equations~\eqref{eq:recursion for extension} and the properties $\pi_N(\tR_n)=R_n$, $\pi_N(\tG_n)=G_n$ and $(\tR_n)^{\le N}_{N+1}=(\tG_n)^{N+1}_{\le N}=0$. Note that the property $[\tQ,\tR_n]=n\tR_n$ implies that the diagonal elements of $\tR_n$ are also zero.

We construct the matrices $\tOmega^d_0$ as a solution of equation~\eqref{eq:TRR for Omega^d_0}, satisfying the initial condition~$\left.\left(\sum_{d\ge 0}\tOmega^{d-1}_0z^d\right)\right|_{t^\alpha=t^\alpha_{\orig}}=(\Id+\sum_{n\ge 1}\tG_n^T z^n)^{-1}$, and one can easily check that $\pi_N(\tOmega^d_0)=\Omega^d_0$ and $\tOmega^{\le N,d}_{N+1,0}=0$. \end{proof}

Let us now apply this lemma to the conformal flat F-manifold, associated to a solution of the open WDVV equations, satisfying property~\eqref{eq:unit condition for Fo} and the homogeneity condition~\eqref{eq:homogeneity for Fo}. We assume that the matrices $\Omega^d_0$ and $R_n$ give a calibration of the Frobenius manifold. By Lemma~\ref{lemma:calibration of extension}, there exists a calibration of the flat F-manifold $M\times U$, given by matrices $\tOmega^d_0$ and~$\tR_n$, satisfying properties~\eqref{eq:property of calibration of extension}. Consider the matrices $\tOmega^0_d$ and define functions $\Phi_{\beta,d}$, $1\le\beta\le N$, $d\ge -1$, by
\begin{gather}\label{eq:Phi and Omega}
\Phi_{\beta,d}:=\tOmega^{N+1,0}_{\beta,d}.
\end{gather}
\begin{lemma}\label{lemma:functions Phi and calibrations}
1. The functions $\Phi_{\beta,d}$ together with the matrices $\tR_n$ give a calibration of the function~$F^o$. As a corollary, Lemma~\ref{lemma:existence of open calibrations} is true.\\
2. Equation~\eqref{eq:Phi and Omega} defines a correspondence between calibrations of the flat F-manifold $M\times U$, satisfying properties~\eqref{eq:property of calibration of extension}, and calibrations of the function $F^o$.
\end{lemma}
\begin{proof}
1. Property~\eqref{eq:open calibration,property 1} is obvious. Equations~\eqref{eq:TRR for Omega^0_d} and~\eqref{eq:homogeneity for Omega^0_d} for the matrices $\tOmega^0_d$ give exactly properties~\eqref{eq:open calibration,property 2} and~\eqref{eq:open calibration, property 3}, respectively.

2. The fact that equations~\eqref{eq:TRR for Omega^0_d} and~\eqref{eq:homogeneity for Omega^0_d} for the matrices $\Omega^0_d$ together with equations~\eqref{eq:open calibration,property 2} and~\eqref{eq:open calibration, property 3} give the required equations for the matrices $\tOmega^0_d$ is obvious. 
\end{proof}

\subsection{Descendent vector potentials of a flat F-manifold}\label{subsection:descendent vector potential}

Consider a flat F-manifold, given by a vector potential $(F^1,\ldots,F^N)$, $F^\alpha\in\mcO(M)$, and let us choose a calibration. One can immediately see that $\frac{\d\Omega^{\alpha,0}_{1,0}}{\d t^\beta}=\delta^\alpha_\beta$, which implies that $\Omega^{\alpha,0}_{1,0}-t^\alpha$ is a constant. Let us make the change of coordinates $t^\alpha\mapsto\Omega^{\alpha,0}_{1,0}$, so that we have now $\Omega^{\alpha,0}_{1,0}=t^\alpha$.

Let $v^1,\ldots,v^N$ be formal variables and consider the {\it principal hierarchy} associated to our flat F-manifold and its calibration (see e.g.~\cite[Section 3.2]{AL18}):
\begin{gather}\label{eq:principal hierarchy for F-man}
\frac{\d v^\alpha}{\d t^\beta_d}=\d_x\left(\left.\Omega^{\alpha,0}_{\beta,d}\right|_{t^\gamma=v^\gamma}\right),\quad 1\le\alpha,\beta\le N,\quad d\ge 0.
\end{gather}
The flows of the principal hierarchy pairwise commute. Since $\Omega^{\alpha,0}_{1,0}=t^\alpha$, we can identify $x=t^1_0$.

Clearly the functions $v^\alpha=t^\alpha_0$ satisfy the subsystem of system~\eqref{eq:principal hierarchy for F-man}, given by the flows $\frac{\d}{\d t^\beta_0}$. Denote by $(v^\top)^\alpha\in\mcO(M)[[t^*_{\ge 1}]]$ the solution of the principal hierarchy specified by the initial condition
$$
\left.(v^\top)^\alpha\right|_{t^*_{\ge 1}=0}=t^\alpha_0.
$$
Recall that $\tt^\alpha_d=t^\alpha_d-\delta^{\alpha,1}\delta_{d,1}$. 
\begin{lemma}
We have
\begin{align}
\sum_{d\ge 0}\tt^\gamma_{d+1}\frac{\d (v^\top)^\alpha}{\d t^\gamma_d}+\delta^{\alpha,1}=&0,\label{eq:string for vtop}\\
\sum_{d\ge 0}\tt^\gamma_d\frac{\d (v^\top)^\alpha}{\d t^\gamma_d}=&0.\label{eq:dilaton for vtop}
\end{align}
\end{lemma}
\begin{proof}
From the property $\frac{\d\Omega^0_d}{\d t^1}=\Omega^0_{d-1}$, $d\ge 0$, it is easy to deduce that the system
$$
\frac{\d v^\alpha}{\d \tau_{-1}}=\sum_{d\ge 0}\tt^\gamma_{d+1}\frac{\d v^\alpha}{\d t^\gamma_d}+\delta^{\alpha,1},\quad 1\le\alpha\le N,
$$
is a symmetry of the principal hierarchy~\eqref{eq:principal hierarchy for F-man}. Since, obviously, $\left.\frac{\d (v^\top)^\alpha}{\d \tau_{-1}}\right|_{t^*_{\ge 1}=0}=0$, we get equation~\eqref{eq:string for vtop}. 

The rescaling combined with the shift along $t^1_1$, given by
$$
\frac{\d v^\alpha}{\d \tau_0}=\sum_{d\ge 0}\tt^\gamma_d\frac{\d v^\alpha}{\d t^\gamma_d},\quad 1\le\alpha\le N,
$$
is also a symmetry of the principal hierarchy. We compute
$$
\left.\frac{\d (v^\top)^\alpha}{\d \tau_0}\right|_{t^*_{\ge 1}}=t^\alpha_0-\left.\frac{\d(v^\top)^\alpha}{\d t^1_1}\right|_{t^*_{\ge 1}=0}=t^\alpha_0-\frac{\d\Omega^{\alpha,0}_{1,1}}{\d t^1_0}=0,
$$
concluding that equation~\eqref{eq:dilaton for vtop} is true.
\end{proof}

Define matrices $\Omega^p_q=(\Omega^{\alpha,p}_{\beta,q})$, $p,q\ge 0$, by
\begin{gather*}
\Omega^p_q:=\sum_{i=0}^q(-1)^{q-i}\Omega^{p+q-i}_0\Omega^0_{i-1}\stackrel{\scriptsize{\text{eq. \eqref{eq:upper-lower relation}}}}{=}\sum_{i=0}^p(-1)^{p-i}\Omega^{i-1}_0\Omega^0_{p+q-i}.
\end{gather*}
We also adopt the convention $\Omega^{-1}_q=\Omega^q_{-1}:=\delta_{q,0}\Id$, $q\ge 0$. One can easily check that
\begin{align}
d\Omega^p_q=&\Omega^{p-1}_0 d\Omega^0_q=d\Omega^p_0\cdot\Omega^0_{q-1},&& p,q\ge 0,\label{eq:TRR for Omega^p_q}\\
\frac{\d\Omega^p_q}{\d t^1}=&\Omega^{p-1}_q+\Omega^p_{q-1},&& p,q\ge 0,\quad p+q\ge 1.\label{eq:string for Omegapq}
\end{align}

Let 
$$
(\Omega^\top)^{\alpha,p}_{\beta,q}:=\left.\Omega^{\alpha,p}_{\beta,q}\right|_{t^\gamma\mapsto(v^\top)^\gamma}.
$$
\begin{lemma}\label{lemma:property of Omegatop}
We have $\frac{\d(\Omega^\top)^{\alpha,p}_{\beta,q}}{\d t^\gamma_r}=\frac{\d(\Omega^\top)^{\alpha,p}_{\gamma,r}}{\d t^\beta_q}$, $p,q,r\ge 0$.
\end{lemma}
\begin{proof}
By~\eqref{eq:TRR for Omega^p_q}, we have $\frac{\d(\Omega^\top)^{\alpha,p}_{\beta,q}}{\d t^\gamma_r}=(\Omega^\top)^{\alpha,p-1}_{\mu,0}\frac{\d(\Omega^\top)^{\mu,0}_{\beta,q}}{\d t^\gamma_r}$. The fact that the flows of the principal hierarchy pairwise commute implies that $\frac{\d(\Omega^\top)^{\mu,0}_{\beta,q}}{\d t^\gamma_r}=\frac{\d(\Omega^\top)^{\mu,0}_{\gamma,r}}{\d t^\beta_q}$. This completes the proof of the lemma.
\end{proof}

We finally define the descendent vector potentials $(\mcF^{1,p},\ldots,\mcF^{N,p})$, $p\ge 0$, associated to our flat F-manifold and its calibration, by
$$
\mcF^{\alpha,p}:=\sum_{q\ge 0}(\Omega^\top)^{\alpha,p}_{\beta,q}\tt^\beta_q\in\mcO(M)[[t^*_{\ge 1}]],\quad p\ge 0.
$$
Let us also adopt the convention
\begin{gather*}
\mcF^{\alpha,p}:=(-1)^{p+1}\tt^\alpha_{-p-1},\quad\text{if $p<0$}.
\end{gather*}

\begin{proposition}\label{proposition:properties of descendent vector potential}
1. We have
\begin{align}
\frac{\d\mcF^{\alpha,p}}{\d t^\beta_q}=&(\Omega^\top)^{\alpha,p}_{\beta,q}, && p\ge 0,\label{eq:derivative of Falphap}\\
\sum_{q\ge 0}\tt^\beta_{q+1}\frac{\d\mcF^{\alpha,p}}{\d t^\beta_q}+\mcF^{\alpha,p-1}=&0, && p\in\mbZ.\label{eq:string for Falphap}
\end{align}
2. The difference $\mcF^{\alpha,0}|_{t^*_{\ge 1}=0}-F^\alpha$ is at most linear in the variables $t^1,\ldots,t^N$.
\end{proposition}
\begin{proof}
1. We compute
\begin{gather*}
\frac{\d\mcF^{\alpha,p}}{\d t^\beta_q}=\sum_{r\ge 0}\frac{\d(\Omega^\top)^{\alpha,p}_{\gamma,r}}{\d t^\beta_q}\tt^\gamma_r+(\Omega^\top)^{\alpha,p}_{\beta,q}\stackrel{\scriptstyle\text{Lemma \ref{lemma:property of Omegatop}}}{=}\sum_{r\ge 0}\frac{\d(\Omega^\top)^{\alpha,p}_{\beta,q}}{\d t^\gamma_r}\tt^\gamma_r+(\Omega^\top)^{\alpha,p}_{\beta,q}\stackrel{\scriptstyle\text{eq. \eqref{eq:dilaton for vtop}}}{=}(\Omega^\top)^{\alpha,p}_{\beta,q}.
\end{gather*}
For $p<0$ equation~\eqref{eq:string for Falphap} is obvious. Suppose $p\ge 0$, then we have
\begin{align*}
\sum_{q\ge 0}\tt^\beta_{q+1}\frac{\d\mcF^{\alpha,p}}{\d t^\beta_q}=&\sum_{q,r\ge 0}\frac{\d(\Omega^\top)^{\alpha,p}_{\gamma,r}}{\d t^\beta_q}\tt^\beta_{q+1}\tt^\gamma_r+\sum_{q\ge 0}\tt^\beta_{q+1}(\Omega^\top)^{\alpha,p}_{\beta,q}\stackrel{\scriptstyle\text{eq. \eqref{eq:string for vtop}}}{=}\\
=&-\sum_{r\ge 0}\frac{\d(\Omega^\top)^{\alpha,p}_{\gamma,r}}{\d (v^\top)^1}\tt^\gamma_r+\sum_{q\ge 0}\tt^\beta_{q+1}(\Omega^\top)^{\alpha,p}_{\beta,q}\stackrel{\scriptstyle\text{eq. \eqref{eq:string for Omegapq}}}{=}\\
=&-\sum_{r\ge 0}(\Omega^\top)^{\alpha,p}_{\gamma,r}\tt^\gamma_{r+1}-\sum_{r\ge 0}(\Omega^\top)^{\alpha,p-1}_{\gamma,r}\tt^\gamma_r+\sum_{q\ge 0}\tt^\beta_{q+1}(\Omega^\top)^{\alpha,p}_{\beta,q}=\\
=&-\mcF^{\alpha,p-1}.
\end{align*}
2. By the first part of the proposition, $\frac{\d}{\d t^\beta_0}\left(\mcF^{\alpha,0}|_{t^*_{\ge 1}=0}\right)=\Omega^{\alpha,0}_{\beta,0}$. Since $\frac{\d\Omega^{\alpha,0}_{\beta,0}}{\d t^\gamma}=c^\alpha_{\beta\gamma}$, the second part of the proposition is also proved.
\end{proof}

\begin{remark}
Consider a Frobenius manifold, its calibration and the associated flat F-manifold. Then the functions $\mcF^{\alpha,p}$ are related to the descendent potential $\mcF$ of the Frobenius manifold by $\mcF^{\alpha,p}=\eta^{\alpha\mu}\frac{\d\mcF}{\d t^\mu_p}$.  
\end{remark}

Consider a conformal Frobenius manifold together with a calibration and a solution $F^o$ of the open WDVV equations, satisfying properties~\eqref{eq:unit condition for Fo},~\eqref{eq:homogeneity for Fo}, also with a calibration. We have the associated flat F-manifold with the vector potential $\left(\eta^{1\mu}\frac{\d F}{\d t^\mu},\ldots,\eta^{N\mu}\frac{\d F}{\d t^\mu},F^o\right)$. Immediately from the definitions and Lemma~\ref{lemma:functions Phi and calibrations} we see that if $(\mcF^{1,p},\ldots,\mcF^{N+1,p})$, $p\ge 0$, are the decendent vector potentials of this flat F-manifold, then $\mcF^{\alpha,p}=\eta^{\alpha\mu}\frac{\d\mcF}{\d t^\mu_p}$, $1\le\alpha\le N$, and $\mcF^{N+1,0}=\mcF^o$. Therefore, Lemma~\ref{lemma:open descendent potential} follows from Proposition~\ref{proposition:properties of descendent vector potential} and equation~\eqref{eq:TRR for Omega^p_q}.

\subsection{Virasoro constraints}

We consider a conformal flat F-manifold, given by a vector potential $(F^1,\ldots,F^N)$, $F^\alpha\in\mcO(M)$, and an Euler vector field~\eqref{eq:Euler vector field}, its calibration, described by matrices $\Omega^d_0$, $d\ge -1$, and $R_n$, $n\ge 1$, and the associated descendent vector potentials $(\mcF^{1,p},\ldots,\mcF^{N,p})$, $p\ge 0$.

Recall that $R=\sum_{i\ge 1}R_i$. Let $\lambda$ be a complex parameter and define
$$
\mu^\alpha:=q^\alpha+\lambda-\frac{3}{2},\qquad \mu:=\diag(\mu^1,\ldots,\mu^N).
$$
Define the following expressions, depending on the parameter $\lambda$:
\begin{align*}
A^\alpha_m:=&\sum_{d_1\ge -1,\,d_2\in\mbZ}(-1)^{d_2+1}\left([P_m(\mu-d_2,R)]_{m-1-d_1-d_2}\right)^\gamma_\mu(\Omega^\top)^{\alpha,0}_{\gamma,d_1}\mcF^{\mu,d_2},\quad 1\le\alpha\le N,\quad m\ge -1.
\end{align*}

\begin{proposition}\label{proposition:Virasoro for F-man}
We have $A^\alpha_m=0$, $1\le\alpha\le N$, $m\ge -1$.
\end{proposition}
\begin{remark}
For a Frobenius manifold the expressions $A^\alpha_m$ have the following interpretation:
$$
A^\alpha_m|_{\lambda=\frac{3-\delta}{2}}=\eta^{\alpha\mu}\frac{\d}{\d t^\mu_0}\left[\Coef_{\eps^{-2}}\left(\frac{L_m e^{\eps^{-2}\mcF}}{e^{\eps^{-2}\mcF}}\right)\right],
$$
where $L_m$ are the Virasoro operators described in Section~\ref{subsection:Virasoro constraints}. Therefore, we consider Proposition~\ref{proposition:Virasoro for F-man} as a generalization of the Virasoro constraints~\eqref{eq:Virasoro constraints} for an arbitrary conformal flat F-manifold.
\end{remark}
\begin{proof}[Proof of Proposition~\ref{proposition:Virasoro for F-man}]
During the proof of the proposition, for the sake of shortness, we will denote the functions~$(v^\top)^\alpha$ and $(\Omega^\top)^{\alpha,p}_{\beta,q}$ by $v^\alpha$ and $\Omega^{\alpha,p}_{\beta,q}$, respectively. We will also denote the function $c^\alpha_{\beta\gamma}|_{t^\theta\mapsto(v^\top)^\theta}$ by $c^\alpha_{\beta\gamma}$ and the function $(1-q^\alpha)(v^\top)^\alpha+r^\alpha$ by~$E^\alpha$.

Define an operator $\mcB=(\mcB^\alpha_\beta)$, depending on the parameter $\lambda$, by
$$
\mcB^\alpha_\beta:=E^\mu c^\alpha_{\mu\beta}\d_x+v^\gamma_x(\lambda-1+q^\beta)c^\alpha_{\gamma\beta}=E^\mu c^\alpha_{\mu\beta}\d_x+v^\gamma_x\left(\mu^\beta+\frac{1}{2}\right)c^\alpha_{\gamma\beta}.
$$
\begin{remark}
For a Frobenius manifold the operator $\mcB|_{\lambda=\frac{3-\delta}{2}}\d_x^{-1}$ coincides with the recursion operator $\mathcal{R}$ from the paper~\cite[equation~(3.36)]{DZ99}.
\end{remark}
We begin with the following lemma.
\begin{lemma}\label{lemma:recursion for Omega^0_d}
We have $\mcB\Omega^0_d=\d_x\left(\Omega^0_{d+1}\left(\mu+d+\frac{3}{2}\right)+\sum_{i=1}^{d+2}\Omega^0_{d+1-i}R_i\right)$, $d\ge -1$.
\end{lemma}
\begin{proof}
Note that $\mcB^\alpha_\beta=\d_x\cdot(E^\mu c^\alpha_{\mu\beta})+v^\gamma_x(\lambda+q^\alpha-2)c^\alpha_{\gamma\beta}$. This follows from the property $E^\mu\frac{\d c^\alpha_{\beta\gamma}}{\d v^\mu}=(q^\beta+q^\gamma-q^\alpha)c^\alpha_{\beta\gamma}$ which is equivalent to the homogeneity condtition~\eqref{eq:homogeneity for F-man}. We then compute $\mcB\Omega^0_d=\d_x\left(E^\mu\frac{\d\Omega^0_{d+1}}{\d v^\mu}+(\lambda+Q-2)\Omega^0_{d+1}\right)$, that, by equation~\eqref{eq:homogeneity for Omega^0_d}, implies the lemma.
\end{proof}

Denote by $A_m$ the column vector $(A^1_m,\ldots,A^N_m)$.
\begin{lemma}\label{lemma:recursion for A_m}
We have 
\begin{gather*}
\mcB A_m-\d_x A_{m+1}=\left(\frac{1}{2}-\mu\right)\sum_{d_1,d_2\ge -1}(-1)^{d_2+1}\Omega^0_{d_1}[P_m(\mu-d_2,R)]_{m-d_1-d_2}\Omega^{d_2}_0\cdot e,\quad m\ge -1,
\end{gather*}
where $e$ denotes the unit vector $(1,0,\ldots,0)$.
\end{lemma}
\begin{proof}
We compute
\begin{align}
\mcB^\alpha_\mu A^\mu_m=&\sum_{d_1\ge -1,\,d_2\in\mbZ}(-1)^{d_2+1}\left([P_m(\mu-d_2,R)]_{m-1-d_1-d_2}\right)^\gamma_\nu\mcB^\alpha_\mu\Omega^{\mu,0}_{\gamma,d_1}\cdot\mcF^{\nu,d_2}\notag\\
&+\sum_{d\ge -1}\left([P_m(\mu+1,R)]_{m-d}\right)^\gamma_1E^\theta\frac{\d\Omega^{\alpha,0}_{\gamma,d+1}}{\d v^\theta}\notag\\
&+\sum_{d_1\ge -1,\,d_2\ge 0}(-1)^{d_2+1}\left([P_m(\mu-d_2,R)]_{m-1-d_1-d_2}\right)^\gamma_\nu E^\theta\frac{\d\Omega^{\alpha,0}_{\gamma,d_1+1}}{\d v^\theta}\Omega^{\nu,d_2}_{1,0}=\notag\\
=&\sum_{d_1\ge -1,\,d_2\in\mbZ}(-1)^{d_2+1}\left([P_m(\mu-d_2,R)]_{m-1-d_1-d_2}\right)^\gamma_\nu\mcB^\alpha_\mu\Omega^{\mu,0}_{\gamma,d_1}\cdot\mcF^{\nu,d_2}\label{eq:Virasoro for F-man.,A1}\\
&+\sum_{d_1,d_2\ge -1}(-1)^{d_2+1}\left([P_m(\mu-d_2,R)]_{m-1-d_1-d_2}\right)^\gamma_\nu E^\theta\frac{\d\Omega^{\alpha,0}_{\gamma,d_1+1}}{\d v^\theta}\Omega^{\nu,d_2}_{1,0}.\label{eq:Virasoro for F-man.,A2}
\end{align}
On the other hand, we have
\begin{align}
\d_x A^\alpha_{m+1}=&\sum_{d_1\ge 0,\,d_2\in\mbZ}(-1)^{d_2+1}\left([P_{m+1}(\mu-d_2,R)]_{m-d_1-d_2}\right)^\gamma_\mu\d_x\Omega^{\alpha,0}_{\gamma,d_1}\cdot\mcF^{\mu,d_2} \label{eq:Virasoro for F-man.,B1}\\
&+\sum_{d_1,d_2\ge -1}(-1)^{d_2+1}\left([P_{m+1}(\mu-d_2,R)]_{m-d_1-d_2}\right)^\gamma_\mu \Omega^{\alpha,0}_{\gamma,d_1}\Omega^{\mu,d_2}_{1,0}. \label{eq:Virasoro for F-man.,B2}
\end{align}

Let us show that the expression in line~\eqref{eq:Virasoro for F-man.,A1} is equal to the expression in line~\eqref{eq:Virasoro for F-man.,B1}. Using Lemma~\ref{lemma:recursion for Omega^0_d}, we rewrite the first one as follows:
\begin{align*}
&\sum_{d_1\ge -1,\,d_2\in\mbZ}(-1)^{d_2+1}\mcF^{\nu,d_2}\left([P_m(\mu-d_2,R)]_{m-1-d_1-d_2}\right)^\gamma_\nu\d_x\left(\left(\mu^\gamma+d_1+\frac{3}{2}\right)\Omega^{\alpha,0}_{\gamma,d_1+1}+\right.\\
&\hspace{11.5cm}\left.+\sum_{i=1}^{d_1+1}\Omega^{\alpha,0}_{\theta,d_1+1-i}(R_i)^\theta_\gamma\right)=\\
=&\sum_{d_1\ge -1,\,d_2\in\mbZ}(-1)^{d_2+1}\mcF^{\nu,d_2}\left(\left[P_m(\mu-d_2,R)\left(\mu-d_2+m+\frac{1}{2}\right)\right]_{m-1-d_1-d_2}\right)^\gamma_\nu\d_x\Omega^{\alpha,0}_{\gamma,d_1+1}\\
&+\sum_{d_1\ge -1,\,d_2\in\mbZ}\sum_{i=1}^{d_1+1}(-1)^{d_2+1}\mcF^{\nu,d_2}([R_iP_m(\mu-d_2,R)]_{m-1-d_1-d_2+i})^\theta_\nu\d_x\Omega^{\alpha,0}_{\theta,d_1+1-i}=\\
=&\sum_{d_1\ge -1,\,d_2\in\mbZ}(-1)^{d_2+1}\mcF^{\nu,d_2}\left(\left[P_m(\mu-d_2,R)\left(\mu-d_2+m+\frac{1}{2}\right)+\right.\right.\\
&\hspace{8.35cm}+R P_m(\mu-d_2,R)\bigg]_{m-1-d_2-d_2}\bigg)^\gamma_\nu\d_x\Omega^{\alpha,0}_{\gamma,d_1+1}=\\
=&\sum_{d_1\ge -1,\,d_2\in\mbZ}(-1)^{d_2+1}\mcF^{\nu,d_2}([P_{m+1}(\mu-d_2,R)]_{m-1-d_1-d_2})^\gamma_\nu\d_x\Omega^{\alpha,0}_{\gamma,d_1+1},
\end{align*}
and we see that the last expression coincides with the expression in line~\eqref{eq:Virasoro for F-man.,B1}.

Using equation~\eqref{eq:homogeneity for Omega^0_d}, we transform the expression in line~\eqref{eq:Virasoro for F-man.,A2}:
\begin{align*}
&\sum_{d_1\ge 0,\,d_2\ge -1}(-1)^{d_2+1}\left([P_m(\mu-d_2,R)]_{m-d_1-d_2}\right)^\gamma_\nu\bigg((\mu^\gamma-\mu^\alpha+d_1+1)\Omega^{\alpha,0}_{\gamma,d_1}+\\
&\hspace{10cm}\left.+\sum_{i=1}^{d_1+1}\Omega^{\alpha,0}_{\theta,d_1-i}(R_i)^\theta_\gamma\right)\Omega^{\nu,d_2}_{1,0}=\\
=&\sum_{d_1\ge 0,\,d_2\ge -1}(-1)^{d_2+1}\left([P_m(\mu-d_2,R)(\mu-\mu^\alpha+m-d_2+1)]_{m-d_1-d_2}\right)^\gamma_\nu\Omega^{\alpha,0}_{\gamma,d_1}\Omega^{\nu,d_2}_{1,0}+\\
&+\sum_{d_1\ge 0,\,d_2\ge -1}\sum_{i=1}^{d_1+1}(-1)^{d_2+1}\left([R_iP_m(\mu-d_2,R)]_{m-d_1-d_2+i}\right)^\gamma_\nu\Omega^{\alpha,0}_{\gamma,d_1-i}\Omega^{\nu,d_2}_{1,0}=\\
=&\left(\frac{1}{2}-\mu^\alpha\right)\sum_{d_1\ge 0,\,d_2\ge -1}(-1)^{d_2+1}\left([P_m(\mu-d_2,R)]_{m-d_1-d_2}\right)^\gamma_\nu\Omega^{\alpha,0}_{\gamma,d_1}\Omega^{\nu,d_2}_{1,0}\\
&+\sum_{d_1\ge 0,\,d_2\ge -1}(-1)^{d_2+1}\left(\left[P_m(\mu-d_2,R)\left(\mu+m-d_2+\frac{1}{2}\right)\right]_{m-d_1-d_2}\right)^\gamma_\nu\Omega^{\alpha,0}_{\gamma,d_1}\Omega^{\nu,d_2}_{1,0}\\
&+\sum_{d_1,d_2\ge -1}(-1)^{d_2+1}\left([RP_m(\mu-d_2,R)]_{m-d_1-d_2}\right)^\gamma_\nu\Omega^{\alpha,0}_{\gamma,d_1}\Omega^{\nu,d_2}_{1,0}.
\end{align*}
On the other hand, the expression in line~\eqref{eq:Virasoro for F-man.,B2} is equal to
\begin{align*}
&\sum_{d_1,d_2\ge -1}(-1)^{d_2+1}\left(\left[P_m(\mu-d_2,R)\left(\mu+m-d_2+\frac{1}{2}\right)\right]_{m-d_1-d_2}\right)^\gamma_\nu\Omega^{\alpha,0}_{\gamma,d_1}\Omega^{\nu,d_2}_{1,0}\\
&+\sum_{d_1,d_2\ge -1}(-1)^{d_2+1}\left([RP_m(\mu-d_2,R)]_{m-d_1-d_2}\right)^\gamma_\nu\Omega^{\alpha,0}_{\gamma,d_1}\Omega^{\nu,d_2}_{1,0}.
\end{align*}

As a result, we get
\begin{align*}
\mcB^\alpha_\mu A^\mu_m-\d_x A^\alpha_{m+1}=&\left(\frac{1}{2}-\mu^\alpha\right)\sum_{d_1\ge 0,\,d_2\ge -1}(-1)^{d_2+1}\left([P_m(\mu-d_2,R)]_{m-d_1-d_2}\right)^\gamma_\nu\Omega^{\alpha,0}_{\gamma,d_1}\Omega^{\nu,d_2}_{1,0}\\
&-\sum_{d_2\ge -1}(-1)^{d_2+1}\left(\left[P_m(\mu-d_2,R)\left(\mu+m-d_2+\frac{1}{2}\right)\right]_{m+1-d_2}\right)^\alpha_\nu\Omega^{\nu,d_2}_{1,0}=\\
=&\left(\frac{1}{2}-\mu^\alpha\right)\sum_{d_1,d_2\ge -1}(-1)^{d_2+1}\left([P_m(\mu-d_2,R)]_{m-d_1-d_2}\right)^\gamma_\nu\Omega^{\alpha,0}_{\gamma,d_1}\Omega^{\nu,d_2}_{1,0},
\end{align*}
as required.
\end{proof}

For $m\ge n\ge -1$ denote
$$
C_{m,n}:=\sum_{d_1,d_2\ge -1}(-1)^{d_2+1}\Omega^0_{d_1}[P_n(\mu-d_2,R)]_{m-d_1-d_2}\Omega^{d_2}_0.
$$
\begin{lemma}\label{lemma:identity for F-man}
We have $C_{m,n}=0$, $m\ge n\ge -1$.
\end{lemma}
\begin{proof}
We proceed by induction on $n$. By equation~\eqref{eq:upper-lower relation}, we have 
$$
C_{m,-1}=\sum_{\substack{d_1,d_2\ge -1\\d_1+d_2=m}}(-1)^{d_2+1}\Omega^0_{d_1}\Omega^{d_2}_0=0,\quad m\ge -1.
$$
Suppose that $m\ge n\ge 0$. We compute
\begin{align*}
&C_{m,n}=\\
=&\sum_{d_1,d_2\ge -1}(-1)^{d_2+1}\Omega^0_{d_1}\left[P_{n-1}(\mu-d_2,R)\left(\mu-d_2+n-\frac{1}{2}\right)+R P_{n-1}(\mu-d_2,R)\right]_{m-d_1-d_2}\Omega^{d_2}_0=\\
=&\sum_{d_1,d_2\ge -1}(-1)^{d_2+1}\Omega^0_{d_1}\left[\left(\mu-m+n+d_1-\frac{1}{2}\right)P_{n-1}(\mu-d_2,R)+R P_{n-1}(\mu-d_2,R)\right]_{m-d_1-d_2}\Omega^{d_2}_0.
\end{align*}
Using property~\eqref{eq:homogeneity for Omega^0_d}, we transform the last expression in the following way:
\begin{align*}
&\left(\mu-m+n-\frac{3}{2}\right)C_{m,n-1}+\sum_{d_1,d_2\ge -1}(-1)^{d_2+1}E^\theta\frac{\d\Omega^0_{d_1}}{\d v^\theta}[P_{n-1}(\mu-d_2,R)]_{m-d_1-d_2}\Omega^{d_2}_0\\
&-\sum_{d_1,d_2\ge -1}\sum_{i=1}^{d_1+1}(-1)^{d_2+1}\Omega^0_{d_1-i}\left[R_iP_{n-1}(\mu-d_2,R)\right]_{m-d_1-d_2+i}\Omega^{d_2}_0\\
&+\sum_{d_1,d_2\ge -1}(-1)^{d_2+1}\Omega^0_{d_1}\left[R P_{n-1}(\mu-d_2,R)\right]_{m-d_1-d_2}\Omega^{d_2}_0.
\end{align*}
One can see that the last two expressions cancel each other and, as a result,
$$
C_{m,n}=\left(\mu-m+n-\frac{3}{2}\right)C_{m,n-1}+E^\theta\frac{\d\Omega^0_0}{\d v^\theta}C_{m-1,n-1},
$$
that, by the induction assumption, is equal to zero. The lemma is proved.
\end{proof}

Lemmas~\ref{lemma:recursion for A_m} and~\ref{lemma:identity for F-man} imply that $\mcB A_m=\d_x A_{m+1}$ for $m\ge -1$. Introduce a differential operator $L:=\sum_{d\ge 0}\tt^\gamma_{d+1}\frac{\d}{\d t^\gamma_d}$.
\begin{lemma}\label{lemma:string for Aalpham}
We have $L A^\alpha_m=(-m-1)A^\alpha_{m-1}$, $m\ge 0$.
\end{lemma}
\begin{proof}
Using formulas~\eqref{eq:string for vtop},~\eqref{eq:string for Falphap} and the formula $\frac{\d\Omega^0_d}{\d v^1}=\Omega^0_{d-1}$, $d\ge 0$, we compute
\begin{align*}
LA^\alpha_m=&-\sum_{d_1\ge 0,\,d_2\in\mbZ}(-1)^{d_2+1}([P_m(\mu-d_2,R)]_{m-1-d_1-d_2})^\gamma_\mu\Omega^{\alpha,0}_{\gamma,d_1-1}\mcF^{\mu,d_2}\\
&-\sum_{d_1\ge -1,\,d_2\in\mbZ}(-1)^{d_2+1}([P_m(\mu-d_2,R)]_{m-1-d_1-d_2})^\gamma_\mu\Omega^{\alpha,0}_{\gamma,d_1}\mcF^{\mu,d_2-1}=\\
=&\sum_{d_1\ge -1,\,d_2\in\mbZ}(-1)^{d_2+1}([P_m(\mu-d_2-1,R)-P_m(\mu-d_2,R)]_{m-2-d_1-d_2})^\gamma_\mu\Omega^{\alpha,0}_{\gamma,d_1}\mcF^{\mu,d_2}.
\end{align*}
Since $P_m(\mu-d_2-1,R)-P_m(\mu-d_2,R)=(-m-1)P_{m-1}(\mu-d_2,R)$, the lemma is proved.
\end{proof}
\begin{lemma}
We have $\left.A^\alpha_m\right|_{t^*_{\ge 1}=0}=0$, $m\ge -1$.
\end{lemma}
\begin{proof}
During the proof of this lemma we return to the initial notation, where $\Omega^{\alpha,p}_{\beta,q}$ is a function of $t^1,\ldots,t^N$. Note that, by equations~\eqref{eq:derivative of Falphap} and~\eqref{eq:string for Falphap}, we have $\mcF^{\alpha,p}|_{t^*_{\ge 1}=0}=\Omega^{\alpha,p+1}_{1,0}$, $p\ge 0$. Therefore, we have
\begin{align*}
\left.A^\alpha_m\right|_{t^*_{\ge 1}=0}=&\sum_{d\ge -1}([P_m(\mu+1,R)]_{m-d})^\gamma_\mu t^\mu_0\Omega^{\alpha,0}_{\gamma,d}\\
&-\sum_{d\ge -1}([P_m(\mu+2,R)]_{m+1-d})^\gamma_1\Omega^{\alpha,0}_{\gamma,d}\\
&+\sum_{d_1\ge -1,\,d_2\ge 0}(-1)^{d_2+1}([P_m(\mu-d_2,R)]_{m-1-d_1-d_2})^\gamma_\mu \Omega^{\alpha,0}_{\gamma,d_1}\Omega^{\mu,d_2+1}_{1,0}=\\
=&-\sum_{d_1,d_2\ge -1}(-1)^{d_2+1}([P_m(\mu+1-d_2,R)]_{m-d_1-d_2})^\gamma_\mu \Omega^{\alpha,0}_{\gamma,d_1}\Omega^{\mu,d_2}_{1,0}=\\
=&-(C_{m,m})^\alpha_1\Big|_{\substack{t^*_{\ge 1}=0\\\lambda\mapsto\lambda+1}},
\end{align*}
which, by Lemma \ref{lemma:identity for F-man}, is equal to zero.
\end{proof}

Let us now prove that $A^\alpha_m=0$ by induction on $m$. We have
\begin{gather*}
A^\alpha_{-1}=\sum_{\substack{d\ge -1,\,d_2\in\mbZ\\d_1+d_2=-2}}(-1)^{d_2+1}\Omega^{\alpha,0}_{\gamma,d_1}\mcF^{\gamma,d_2}=t^\alpha_0+\sum_{d\ge 0}\Omega^{\alpha,0}_{\gamma,d}\tt^\alpha_{d+1}=t^\alpha_0+L\mcF^{\alpha,0}\stackrel{\scriptstyle\text{eq. \eqref{eq:string for Falphap}}}{=}0.
\end{gather*}
Suppose that $m\ge 0$. Let us express $A^\alpha_m$ as a power series in the variables $\hatt^\beta_d$ defined by
$$
\hatt^\beta_d:=
\begin{cases}
t^\beta_0-t^\beta_{\orig},&\text{if $d=0$},\\
t^\beta_d,&\text{if $d\ge 1$}.
\end{cases}
$$
From the induction assumption, Lemma~\ref{lemma:string for Aalpham} and the fact that $\mcB A_{m-1}=\d_x A_m$ it follows that
$$
\sum_{d\ge 0}\hatt^\gamma_{d+1}\frac{\d A^\alpha_m}{\d\hatt^\gamma_d}=0.
$$
We also know that $A^\alpha_m|_{\hatt^*_*=0}=0$. By~\cite[Lemma 3.1]{Get99}, this implies that $A^\alpha_m=0$. 
\end{proof}
\begin{remark}
Ezra Getzler has informed us that the proposition can be also proved using the results and the arguments from the papers~\cite{Get99,Get04}.
\end{remark}

\subsection{Proof of Theorem~\ref{theorem:open Virasoro}}\label{subsection:proof of open Virasoro}

Let us apply Proposition~\ref{proposition:Virasoro for F-man} to the flat F-manifold associated to the function $F^o$. We choose $\lambda=\frac{3-\delta}{2}$, then we have
\begin{gather*}
A^{N+1}_m=\sum_{d_1\ge -1,\,d_2\in\mbZ}(-1)^{d_2+1}\left([P_m(\tmu-d_2,\tR)]_{m-1-d_1-d_2}\right)^\gamma_\mu\Phi^\top_{\gamma,d_1}\mcF^{\mu,d_2}.
\end{gather*}
Since $\tR^\alpha_{N+1}=0$, we have $\left([P_m(\tmu-d_2,\tR)]_{m-1-d_1-d_2}\right)^\gamma_{N+1}=0$ for $1\le\gamma\le N$. Therefore,
\begin{align*}
A^{N+1}_m=&\sum_{d_1\ge 0,\,d_2\in\mbZ}\sum_{1\le\gamma,\mu\le N}(-1)^{d_2+1}\left([P_m(\tmu-d_2,\tR)]_{m-1-d_1-d_2}\right)^\gamma_\mu\frac{\d\mcF^o}{\d t^\gamma_{d_1}}\mcF^{\mu,d_2}+\\
&+\sum_{d\ge 0}\left([P_m(\tmu+d+1,\tR)]_{m+d+1}\right)^{N+1}_\alpha\tt^\alpha_d+\\
&+\sum_{d,k\ge 0}\left([P_m(\tmu+d+1,\tR)]_{m+d-k}\right)^{N+1}_\alpha\tt^\alpha_d\frac{\d\mcF^o}{\d s_k}+
\end{align*}
\begin{align*}
&+\sum_{k\ge 0}\sum_{1\le\alpha,\mu\le N}(-1)^{k+1}\left([P_m(\tmu-k,\tR)]_{m-k}\right)^{N+1}_\alpha\eta^{\alpha\mu}\frac{\d\mcF}{\d t^\mu_k}+\\
&+\sum_{d_1,d_2\ge 0}\sum_{1\le\alpha,\mu\le N}(-1)^{d_2+1}\left([P_m(\tmu-d_2,\tR)]_{m-1-d_1-d_2}\right)^{N+1}_\alpha\frac{\d\mcF^o}{\d s_{d_1}}\eta^{\alpha\mu}\frac{\d\mcF}{\d t^\mu_{d_2}}+\\
&+\sum_{d_1\ge -1,\,d_2\ge 0}(-1)^{d_2+1}\left([P_m(\tmu-d_2,\tR)]_{m-1-d_1-d_2}\right)^{N+1}_{N+1}\Phi^\top_{N+1,d_1}\mcF^{N+1,d_2}.
\end{align*}
The first term here is equal to $\Coef_{\eps^{-1}}\left(\frac{L_m e^{\eps^{-2}\mcF+\eps^{-1}\mcF^o}}{e^{\eps^{-2}\mcF+\eps^{-1}\mcF^o}}\right)$. The next four terms correspond to the four summations in the expression~\eqref{eq:exact open Virasoro} for the operators $\mcL_m$. Since $\mu^{N+1}=\frac{1}{2}$, the last term is equal to zero. Thus,
$$
A^{N+1}_m=\Coef_{\eps^{-1}}\left(\frac{\cL_m e^{\eps^{-2}\mcF+\eps^{-1}\mcF^o}}{e^{\eps^{-2}\mcF+\eps^{-1}\mcF^o}}\right)=0,
$$
that proves Theorem~\ref{theorem:open Virasoro}.


\section{Open Virasoro constraints in all genera}\label{section:open Virasoro in all genera}

There is a canonical construction that associates to a given semisimple conformal Frobenius manifold and its calibration a sequence of functions $\mcF_0(t^*_*)=\mcF,\mcF_1(t^*_*),\mcF_2(t^*_*),\ldots$, such that for the differential operators $L_m$, given by~\eqref{eq:Virasoro1}, equations~\eqref{eq:Virasoro in all genera} hold~\cite{Giv01,Giv04,Tel12}. If one considers the Gromov--Witten theory of a given target variety, then the functions $\mcF_g(t^*_*)$ are the generating series of intersection numbers on the moduli space of maps from a Riemann surface of genus $g$ to the target variety.   

We conjecture that, under possibly some additional assumptions, there is a canonical way to associate to a solution~$F^o$ of the WDVV equations, satisfying properties~\eqref{eq:unit condition for Fo},~\eqref{eq:homogeneity for Fo}, and its calibration a sequence of functions $\mcF^o_0(t^{\le N}_*,s_*)=\mcF^o,\mcF_1(t^{\le N}_*,s_*),\mcF_2(t^{\le N}_*,s_*),\ldots$, such that for the differential operators $\mcL_m$, given by~\eqref{eq:exact open Virasoro}, the equations
$$
\mcL_m e^{\sum_{g\ge 0}\eps^{2g-2}\mcF_g+\sum_{g\ge 0}\eps^{g-1}\mcF^o_g}=0,\quad m\ge -1,
$$
are satisfied. At the moment the conjecture is verified only in the case, corresponding to the intersection theory on the moduli space of Riemann surfaces with boundary~\cite{PST14,BT17}.

As a step towards the proof of this conjecture, we verify the following commutation relations between the operators $\mcL_m$. 

\begin{proposition}
We have $[\mcL_m,\mcL_n]=(m-n)\mcL_{m+n}$, $m,n\ge -1$.
\end{proposition}
\begin{proof}
Denote the parts of the expression for the operator $L_m$ from lines~\eqref{eq:Virasoro1},~\eqref{eq:Virasoro2} and~\eqref{eq:Virasoro3} by $L_m^1$, $L_m^2$ and $L_m^3$, correspondingly. One can see that the operators $\mcL_m$ can be written in the following way:
$$
\mcL_m=L_m+\mcL_m^1+\mcL_m^2+\mcL_m^3+\mcL_m^4+\eps^{-1}\delta_{m,-1}s+\delta_{m,0}\frac{3}{4}+\sum_{d\ge 0}\frac{(d+m+1)!}{d!}s_d\frac{\d}{\d s_{d+m}}+\eps\frac{3(m+1)!}{4}\frac{\d}{\d s_{m-1}},
$$
where
\begin{align*}
&\cL_m^1=\eps^{-1}\sum_{d\ge 0}\sum_{\substack{1\le\alpha\le N\\\mu^\alpha=-\frac{1}{2}-m-d}}\left(\prod_{i=0}^m(\tR-i)\right)^{N+1}_\alpha\tt^\alpha_d,\\
&\mcL_m^2=\sum_{d,k\ge 0}\sum_{\substack{1\le\alpha\le N\\\mu^\alpha=\frac{1}{2}-m-d+k}}\left(\prod_{i=0}^m(\tR+k+1-i)\right)^{N+1}_\alpha\tt^\alpha_d\frac{\d}{\d s_k},
\end{align*}
\begin{align*}
&\mcL_m^3=\eps\sum_{k\ge 0}\sum_{\substack{1\le\alpha,\mu\le N\\\mu^\alpha=\frac{1}{2}-m+k}}(-1)^{k+1}\left(\prod_{i=0}^m(\tR-i)\right)^{N+1}_\alpha\eta^{\alpha\mu}\frac{\d}{\d t^\mu_k},\\
&\mcL_m^4=\eps^2\sum_{d_1,d_2\ge 0}\sum_{\substack{1\le\alpha,\mu\le N\\\mu^\alpha=\frac{3}{2}-m+d_1+d_2}}(-1)^{d_2+1}\left(\prod_{i=0}^m(\tR+d_1+1-i)\right)^{N+1}_\alpha\eta^{\alpha\mu}\frac{\d^2}{\d s_{d_1}\d t^\mu_{d_2}}.
\end{align*}

Let us first prove that $[\mcL_{-1},\mcL_n]=(-1-n)\mcL_{n-1}$, for $n\ge 0$. For this we compute
\begin{align}
[\mcL_{-1},L_n]=&(-1-n)L_{n-1},\notag\\
[\mcL_{-1},\mcL^1_n]=&\eps^{-1}\sum_{d\ge 1}\sum_{\substack{1\le\alpha\le N\\\mu^\alpha=\frac{1}{2}-n-d}}\left(\prod_{i=0}^n(\tR-i)\right)^{N+1}_\alpha\tt^\alpha_d,\label{A11}\\
[\mcL_{-1},\mcL^2_n]=&\sum_{d\ge 1,\,k\ge 0}\sum_{\substack{1\le\alpha\le N\\\mu^\alpha=\frac{3}{2}-n-d+k}}\left(\prod_{i=0}^n(\tR+k+1-i)\right)^{N+1}_\alpha\tt^\alpha_d\frac{\d}{\d s_k}\label{A21}\\
&-\eps^{-1}\sum_{d\ge 0}\sum_{\substack{1\le\alpha\le N\\\mu^\alpha=\frac{1}{2}-n-d}}\left(\prod_{i=0}^n(\tR+1-i)\right)^{N+1}_\alpha\tt^\alpha_d\label{A12}\\
&-\sum_{d,k\ge 0}\sum_{\substack{1\le\alpha\le N\\\mu^\alpha=\frac{3}{2}-n-d+k}}\left(\prod_{i=0}^n(\tR+k+2-i)\right)^{N+1}_\alpha\tt^\alpha_d\frac{\d}{\d s_k},\label{A22}\\
[\mcL_{-1},\mcL^3_n]=&\eps^{-1}\sum_{\substack{1\le\alpha\le N\\\mu^\alpha=\frac{1}{2}-n}}\left(\prod_{i=0}^n(\tR-i)\right)^{N+1}_\alpha\tt^\alpha_0\label{A13}\\
&+\eps\sum_{k\ge 0}\sum_{\substack{1\le\alpha,\mu\le N\\\mu^\alpha=\frac{3}{2}-n+k}}(-1)^{k+1}\left(\prod_{i=0}^n(\tR-i)\right)^{N+1}_\alpha\eta^{\alpha\mu}\frac{\d}{\d t^\mu_k},\label{A31}\\
[\mcL_{-1},\mcL^4_n]=&\eps^2\sum_{d_1,d_2\ge 0}\sum_{\substack{1\le\alpha,\mu\le N\\\mu^\alpha=\frac{5}{2}-n+d_1+d_2}}(-1)^{d_2+1}\left(\prod_{i=0}^n(\tR+d_1+1-i)\right)^{N+1}_\alpha\eta^{\alpha\mu}\frac{\d^2}{\d s_{d_1}\d t^\mu_{d_2}}\label{A41}\\
&+\sum_{d\ge 0}\sum_{\substack{1\le\alpha\le N\\\mu^\alpha=\frac{3}{2}-n+d}}\left(\prod_{i=0}^n(\tR+d+1-i)\right)^{N+1}_\alpha\tt^\alpha_0\frac{\d}{\d s_d}\label{A23}\\
&-\eps\sum_{d\ge 0}\sum_{\substack{1\le\alpha,\mu\le N\\\mu^\alpha=\frac{3}{2}-n+d}}(-1)^{d+1}\left(\prod_{i=0}^n(\tR+1-i)\right)^{N+1}_\alpha\eta^{\alpha\mu}\frac{\d}{\d t^\mu_d}\label{A32}\\
&-\eps^2\sum_{d_1,d_2\ge 0}\sum_{\substack{1\le\alpha,\mu\le N\\\mu^\alpha=\frac{5}{2}-n+d_1+d_2}}(-1)^{d_2+1}\left(\prod_{i=0}^n(\tR+d_1+2-i)\right)^{N+1}_\alpha\eta^{\alpha\mu}\frac{\d^2}{\d s_{d_1}\d t^\mu_{d_2}},\label{A42}
\end{align}
and
\begin{align*}
\left[\mcL_{-1},\sum_{d\ge 0}\frac{(d+n+1)!}{d!}s_d\frac{\d}{\d s_{d+n}}\right]=&-\delta_{n,0}\eps^{-1}s+(-1-n)\sum_{d\ge 0}\frac{(d+n)!}{d!}s_d\frac{\d}{\d s_{d+n-1}},\\
\left[\mcL_{-1},\eps\frac{3(n+1)!}{4}\frac{\d}{\d s_{n-1}}\right]=&-\delta_{n,1}\frac{3}{2}-\eps\frac{3(n+1)!}{4}\frac{\d}{\d s_{n-2}}.
\end{align*}
It remains to note that the sum of the expressions in lines~\eqref{A11},~\eqref{A12},~\eqref{A13} is equal to ${(-1-n)\mcL^1_{n-1}}$, the sum of the expressions in lines~\eqref{A21},~\eqref{A22},~\eqref{A23} is equal to $(-1-n)\mcL^2_{n-1}$, the sum of the expressions in lines~\eqref{A31},~\eqref{A32} is equal to $(-1-n)\mcL^3_{n-1}$ and the sum of the expressions in lines~\eqref{A41},~\eqref{A42} is equal to $(-1-n)\mcL^4_{n-1}$. 

Let us now prove the proposition for $m,n\ge 0$. The commutator $[\mcL_m,\mcL_n]$ has the form
\begin{align*}
[\mcL_m,\mcL_n]=(m-n)L_{m+n}&+\eps^{-1}\sum_{d\ge 0}\sum_{1\le\alpha\le N}A_{\alpha,d}\tt^\alpha_d+\\
&+\left(\sum_{p,q\ge 0}\sum_{1\le\alpha\le N}B_{\alpha,p}^q\tt^\alpha_p\frac{\d}{\d s_q}+\sum_{p,q\ge 0}C_p^q s_p\frac{\d}{\d s_q}+D\right)+\\
&+\eps\left(\sum_{d\ge 0}\sum_{1\le\alpha\le N}E^{\alpha,d}\frac{\d}{\d t^\alpha_d}+\sum_{d\ge 0}G^d \frac{\d}{\d s_d}\right)+\\
&+\eps^2\left(\sum_{p,q\ge 0}\sum_{1\le\alpha\le N}H^{\alpha,p;q}\frac{\d}{\d t^\alpha_p\d s_q}+\sum_{p,q\ge 0}I^{p;q}\frac{\d}{\d s_p\d s_q}\right).
\end{align*}
Let us consider separately the terms on the right-hand side of this expression.

1. {\it Term $\eps^{-1}\sum_{d\ge 0}\sum_{1\le\alpha\le N}A_{\alpha,d}\tt^\alpha_d$:}
$$
\eps^{-1}\sum_{d\ge 0}\sum_{1\le\alpha\le N}A_{\alpha,d}\tt^\alpha_d=[L_m^2,\mcL_n^1]+[L_m^1,\mcL_n^3]-[L_n^2,\mcL_m^1]-[L_n^1,\mcL_m^3].
$$
We compute
$$
[L_m^2,\mcL_n^1]=\eps^{-1}\sum_{d,k\ge 0}\sum_{\substack{1\le\alpha\le N\\\mu^\alpha=-\frac{1}{2}-m-n-d}}\sum_{\substack{1\le\beta\le N\\\mu^\beta=-\frac{1}{2}-n-k}}\left(\prod_{i=0}^m(R-n-i)\right)^\beta_\alpha\left(\prod_{i=0}^n(\tR-i)\right)^{N+1}_\beta\tt^\alpha_d.
$$
Similarly, we get
$$
[L_m^1,\mcL_n^3]=\eps^{-1}\sum_{d\ge 0,\,k<0}\sum_{\substack{1\le\alpha\le N\\\mu^\alpha=-\frac{1}{2}-m-n-d}}\sum_{\substack{1\le\beta\le N\\\mu^\beta=-\frac{1}{2}-n-k}}\left(\prod_{i=0}^m(R-n-i)\right)^\beta_\alpha\left(\prod_{i=0}^n(\tR-i)\right)^{N+1}_\beta\tt^\alpha_d.
$$
As a result, 
$$
[L_m^2,\mcL_n^1]+[L_m^1,\mcL_n^3]=\eps^{-1}\sum_{d\ge 0}\sum_{\substack{1\le\alpha\le N\\\mu^\alpha=-\frac{1}{2}-m-n-d}}\left((\tR-n)\prod_{i=0}^{m+n}(\tR-i)\right)^{N+1}_\alpha\tt^\alpha_d,
$$
which finally gives 
\begin{multline*}
[L_m^2,\mcL_n^1]+[L_m^1,\mcL_n^3]-[L_n^2,\mcL_m^1]-[L_n^1,\mcL_m^3]=\\
=(m-n)\eps^{-1}\sum_{d\ge 0}\sum_{\substack{1\le\alpha\le N\\\mu^\alpha=-\frac{1}{2}-m-n-d}}\left(\prod_{i=0}^{m+n}(\tR-i)\right)^{N+1}_\alpha\tt^\alpha_d=(m-n)\mcL_{m+n}^1,
\end{multline*}
as required.

2. {\it Term $\sum_{p,q\ge 0}\sum_{1\le\alpha\le N}B_{\alpha,p}^q\tt^\alpha_p\frac{\d}{\d s_q}$:}
\begin{align*}
\sum_{p,q\ge 0}\sum_{1\le\alpha\le N}B_{\alpha,p}^q\tt^\alpha_p\frac{\d}{\d s_q}=&[L_m^2,\mcL_n^2]+\left[\mcL_m^2,\sum_{d\ge 0}\frac{(d+n+1)!}{d!}s_d\frac{\d}{\d s_{d+n}}\right]+[L_m^1,\mcL_n^4]-\\
&-[L_n^2,\mcL_m^2]-\left[\mcL_n^2,\sum_{d\ge 0}\frac{(d+m+1)!}{d!}s_d\frac{\d}{\d s_{d+m}}\right]-[L_n^1,\mcL_m^4].
\end{align*}
We compute
\begin{align*}
&[L_m^2,\mcL_n^2]+\left[\mcL_m^2,\sum_{d\ge 0}\frac{(d+n+1)!}{d!}s_d\frac{\d}{\d s_{d+n}}\right]=\\
=&\sum_{d,k,l\ge 0}\sum_{\substack{1\le\alpha\le N\\\mu^\alpha=\frac{1}{2}-m-n-d+k}}\sum_{\substack{1\le\beta\le N+1\\\mu^\beta=\frac{1}{2}-n+k-l}}\left(\prod_{i=0}^m(\tR+k+1-n-i)\right)^\beta_\alpha\left(\prod_{i=0}^n(\tR+k+1-i)\right)^{N+1}_\beta\tt^\alpha_d\frac{\d}{\d s_k},
\end{align*}
and then check that this sum, after replacing the summation $\sum_{l\ge 0}$ by the summation $\sum_{l<0}$, is equal to the commutator $[L_m^1,\mcL_n^4]$. This gives
\begin{multline*}
[L_m^2,\mcL_n^2]+\left[\mcL_m^2,\sum_{d\ge 0}\frac{(d+n+1)!}{d!}s_d\frac{\d}{\d s_{d+n}}\right]+[L_m^1,\mcL_n^4]=\\
=\sum_{d,k\ge 0}\sum_{\substack{1\le\alpha\le N\\\mu^\alpha=\frac{1}{2}-m-n-d+k}}\left((\tR+k+1-n)\prod_{i=0}^{m+n}(\tR+k+1-i)\right)^{N+1}_\alpha\tt^\alpha_d\frac{\d}{\d s_k}.
\end{multline*}
As a result, $\sum_{p,q\ge 0}\sum_{1\le\alpha\le N}B_{\alpha,p}^q\tt^\alpha_p\frac{\d}{\d s_q}=(m-n)\mcL_{m+n}^2$, as required.

3. {\it Term $\sum_{p,q\ge 0}C_p^q s_p\frac{\d}{\d s_q}$:}
\begin{align*}
\sum_{p,q\ge 0}C_p^q s_p\frac{\d}{\d s_q}=&\left[\sum_{p\ge 0}\frac{(p+m+1)!}{p!}s_p\frac{\d}{\d s_{p+m}},\sum_{q\ge 0}\frac{(q+n+1)!}{q!}s_q\frac{\d}{\d s_{q+n}}\right]=\\
=&(m-n)\sum_{d\ge 0}\frac{(d+m+n+1)!}{d!}s_d\frac{\d}{\d s_{d+m+n}},
\end{align*}
as required.

4. {\it Constant $D$:}
\begin{align*}
D=&[\mcL_m^1,\mcL_n^3]-[\mcL_n^1,\mcL_m^3],\quad \text{where}\\
[\mcL_m^1,\mcL_n^3]=&\sum_{d\ge 0}\sum_{\substack{1\le\alpha,\beta\le N\\\mu^\alpha=-\frac{1}{2}-m-d\\\mu^\beta=\frac{1}{2}-n+d}}\left(\prod_{i=0}^m(\tR-i)\right)_\alpha^{N+1}\eta^{\alpha\beta}\left(\prod_{i=0}^n(\tR-i)\right)_\beta^{N+1}.
\end{align*}
Because of the property $\mu\eta+\eta\mu=0$, the last expression is equal to zero unless $m+n=0$, which implies that $m=n=0$. Thus, $D=0$, as required.

5. {\it Term $\eps\sum_{d\ge 0}\sum_{1\le\alpha\le N}E^{\alpha,d}\frac{\d}{\d t^\alpha_d}$:} 
\begin{gather*}
\eps\sum_{d\ge 0}\sum_{1\le\alpha\le N}E^{\alpha,d}\frac{\d}{\d t^\alpha_d}=[L_m^2,\mcL_n^3]+[L_m^3,\mcL_n^1]-[L_n^2,\mcL_m^3]-[L_n^3,\mcL_m^1].
\end{gather*}
We first compute
\begin{align}
[L_m^2,\mcL_n^3]=&\eps\sum_{d,k\ge 0}\sum_{\substack{1\le\alpha,\beta,\gamma\le N\\\mu^\beta=m+n-k-\frac{1}{2}\\\mu^\alpha=-\frac{1}{2}+n-d}}(-1)^d\left(\prod_{i=0}^n(\tR-i)\right)^{N+1}_\gamma\eta^{\gamma\alpha}\left(\prod_{i=0}^m(R+n+i)\right)^\beta_\alpha\frac{\d}{\d t^\beta_k}=\notag\\
=&\eps\sum_{d,k\ge 0}\sum_{\substack{1\le\alpha,\beta,\gamma\le N\\\mu^\beta=m+n-k-\frac{1}{2}\\\mu^\gamma=\frac{1}{2}-n+d}}(-1)^{k+1}\left(\prod_{i=0}^n(\tR-i)\right)^{N+1}_\gamma\left(\prod_{i=0}^m(R-n-i)\right)^\gamma_\alpha\eta^{\alpha\beta}\frac{\d}{\d t^\beta_k},\label{eq:commutation,fifth term,1}
\end{align}
where the second equality is obtained using the property~$\eta R_i\eta^{-1}=(-1)^{i-1}R_i^T$. Then one can compute that the commutator $[L_m^3,\mcL_n^1]$ is equal to the expression in line~\eqref{eq:commutation,fifth term,1} with the summation $\sum_{d\ge 0}$ replaced by $\sum_{d<0}$. This implies that
$$
[L_m^2,\mcL_n^3]+[L_m^3,\mcL_n^1]=\eps\sum_{k\ge 0}\sum_{\substack{1\le\alpha,\beta\le N\\\mu^\beta=m+n-k-\frac{1}{2}}}(-1)^{k+1}\left(\prod_{i=0}^n(\tR-i)\prod_{i=0}^m(\tR-n-i)\right)^{N+1}_\alpha\eta^{\alpha\beta}\frac{\d}{\d t^\beta_k},
$$
and, as a result,
\begin{multline*}
[L_m^2,\mcL_n^3]+[L_m^3,\mcL_n^1]-[L_n^2,\mcL_m^3]-[L_n^3,\mcL_m^1]=\\
=(m-n)\eps\sum_{k\ge 0}\sum_{\substack{1\le\alpha,\beta\le N\\\mu^\beta=m+n-k-\frac{1}{2}}}(-1)^{k+1}\left(\prod_{i=0}^{m+n}(\tR-i)\right)^{N+1}_\alpha\eta^{\alpha\beta}\frac{\d}{\d t^\beta_k}=(m-n)\mcL_{m+n}^3,
\end{multline*}
as required.

6. {\it Term $\eps\sum_{d\ge 0}G^d\frac{\d}{\d s_d}$:}
\begin{gather*}
\eps\sum_{d\ge 0}G^d\frac{\d}{\d s_d}=(m-n)\eps\frac{3(m+n+1)!}{4}\frac{\d}{\d s_{m+n-1}}+[\mcL_m^2,\mcL_n^3]+[\mcL_m^1,\mcL_n^4]-[\mcL_n^2,\mcL_m^3]-[\mcL_n^1,\mcL_m^4].
\end{gather*}
We compute
\begin{gather*}
[\mcL_m^2,\mcL_n^3]=\eps\sum_{d\ge 0}(-1)^d\sum_{\substack{1\le\alpha,\beta\le N\\\mu^\alpha=-\frac{1}{2}+n-d\\\mu^\beta=\frac{1}{2}-n+d}}\left(\prod_{i=0}^m(\tR+m+n-i)\right)^{N+1}_\alpha\eta^{\alpha\beta}\left(\prod_{i=0}^n(\tR-i)\right)^{N+1}_\beta\frac{\d}{\d s_{m+n-1}}.
\end{gather*}
A term in this sum is equal to zero unless $\mu^\alpha<\frac{1}{2}\Leftrightarrow d>n-1$ and $\mu^\beta<\frac{1}{2}\Leftrightarrow d<n$, that never happens. Therefore, $[\mcL_m^2,\mcL_n^3]=0$, and, similarly, one can check that $[\mcL_m^1,\mcL_n^4]=[\mcL_n^2,\mcL_m^3]=[\mcL_n^1,\mcL_m^4]=0$. Hence, $\eps\sum_{d\ge 0}G^d\frac{\d}{\d s_d}=(m-n)\eps\frac{3(m+n+1)!}{4}\frac{\d}{\d s_{m+n-1}}$, as required.

7. {\it Term $\eps^2\sum_{p,q\ge 0}\sum_{1\le\alpha\le N}H^{\alpha,p;q}\frac{\d}{\d t^\alpha_p\d s_q}$:}
\begin{align*}
\eps^2\sum_{p,q\ge 0}\sum_{1\le\alpha\le N}H^{\alpha,p;q}\frac{\d}{\d t^\alpha_p\d s_q}=&[L_m^2,\mcL_n^4]+[L_m^3,\mcL_n^2]+\left[\mcL_m^4,\sum_{d\ge 0}\frac{(d+n+1)!}{d!}s_d\frac{\d}{\d s_{d+n}}\right]\\
&-[L_n^2,\mcL_m^4]-[L_n^3,\mcL_m^2]-\left[\mcL_n^4,\sum_{d\ge 0}\frac{(d+m+1)!}{d!}s_d\frac{\d}{\d s_{d+m}}\right].
\end{align*}
We proceed with the computation
\begin{align*}
&[L_m^2,\mcL_n^4]=\\
=&\eps^2\hspace{-0.2cm}\sum_{d_1,d_2,d\ge 0}(-1)^d\hspace{-0.8cm}\sum_{\substack{1\le\alpha,\beta,\mu\le N\\\mu^\alpha=\frac{3}{2}-n+d_1+d\\\mu^\beta=-\frac{3}{2}+m+n-d_1-d_2}}\hspace{-0.6cm}\left(\prod_{i=0}^n(\tR+d_1+1-i)\right)^{N+1}_\alpha\hspace{-0.3cm}\eta^{\alpha\mu}\left(\prod_{i=0}^m(R-1-d_1+n+i)\right)^\beta_\mu\frac{\d^2}{\d s_{d_1}\d t^\beta_{d_2}}.
\end{align*}
Using the relation $\eta R_i\eta^{-1}=(-1)^{i-1}R_i^T$, we convert this sum to
\begin{gather}\label{eq:proof of the commutation,eq1}
\eps^2\hspace{-0.2cm}\sum_{d_1,d_2,d\ge 0}(-1)^{d_2+1}\hspace{-0.8cm}\sum_{\substack{1\le\alpha\le N+1\\\mu^\alpha=\frac{3}{2}-n+d_1+d\\\mu^\beta=-\frac{3}{2}+m+n-d_1-d_2}}\hspace{-0.6cm}\left(\prod_{i=0}^n(\tR+d_1+1-i)\right)^{N+1}_\alpha\left(\prod_{i=0}^m(\tR+1+d_1-n-i)\right)^\alpha_\mu\eta^{\mu\beta}\frac{\d^2}{\d s_{d_1}\d t^\beta_{d_2}}.
\end{gather}
Then one can check that the expression $[L_m^3,\mcL_n^2]+\left[\mcL_m^4,\sum_{d\ge 0}\frac{(d+n+1)!}{d!}s_d\frac{\d}{\d s_{d+n}}\right]$ is equal to the expression~\eqref{eq:proof of the commutation,eq1}, with the summation $\sum_{d\ge 0}$ replaced by $\sum_{d<0}$. As a result,
\begin{align*}
&[L_m^2,\mcL_n^4]+[L_m^3,\mcL_n^2]+\left[\mcL_m^4,\sum_{d\ge 0}\frac{(d+n+1)!}{d!}s_d\frac{\d}{\d s_{d+n}}\right]=\\
=&\eps^2\sum_{d_1,d_2\ge 0}(-1)^{d_2+1}\sum_{\mu^\beta=-\frac{3}{2}+m+n-d_1-d_2}\left((\tR+d_1+1-n)\prod_{i=0}^{m+n}(\tR+d_1+1-i)\right)^{N+1}_\mu\eta^{\mu\beta}\frac{\d^2}{\d s_{d_1}\d t^\beta_{d_2}},
\end{align*}
which gives $\eps^2\sum_{p,q\ge 0}\sum_{1\le\alpha\le N}H^{\alpha,p;q}\frac{\d}{\d t^\alpha_p\d s_q}=(m-n)\mcL_{m+n}^4$, as required.

8. {\it Term $\eps^2\sum_{p,q\ge 0}I^{p;q}\frac{\d^2}{\d s_p\d s_q}$:}
$$
\eps^2\sum_{p,q\ge 0}I^{p;q}\frac{\d^2}{\d s_p\d s_q}=[\mcL_m^4,\mcL_n^2]-[\mcL_n^4,\mcL_m^2],
$$
where
\begin{gather*}
[\mcL_m^4,\mcL_n^2]=\eps^2\hspace{-0.1cm}\sum_{p,q,d\ge 0}\hspace{-0.2cm}\sum_{\substack{1\le\alpha,\beta\le N\\\mu^\alpha=\frac{3}{2}-m+p+d\\\mu^\beta=\frac{1}{2}-n+q-d}}\hspace{-0.4cm}(-1)^{d+1}\left(\prod_{i=0}^m(\tR+p+1-i)\right)^{N+1}_\alpha\hspace{-0.3cm}\eta^{\alpha\beta}\left(\prod_{i=0}^n(\tR+q+1-i)\right)^{N+1}_\beta\hspace{-0.3cm}\frac{\d^2}{\d s_p\d s_q}.
\end{gather*}
A term in this sum is equal to zero unless $\mu^\alpha<\frac{1}{2}\Leftrightarrow p+d\le m-2$, $\mu^\beta<\frac{1}{2}\Leftrightarrow q-d\le n-1$ and $\mu^\alpha+\mu^\beta=0\Leftrightarrow p+q=m+n-2$. The first two condition give $p+q\le m+n-3$, that contradicts the third condition. Thus, $\eps^2\sum_{p,q\ge 0}I^{p;q}\frac{\d^2}{\d s_p\d s_q}=0$, as required. This completes the proof of the proposition.
\end{proof}


\section{Examples of solutions of the open WDVV equations}\label{section: examples}

In this section we present several examples of solutions of the open WDVV equations, satisfying conditions~\eqref{eq:unit condition for Fo} and~\eqref{eq:homogeneity for Fo} and, thus, Theorem~\ref{theorem:open Virasoro} can be applied to them.

\subsection{Extended $r$-spin theory}

Let us fix an integer $r\ge 2$. There is a conformal Frobenius manifold that controls the integrals of the so-called Witten class over the moduli space of stable curves of genus $0$ with an $r$-spin structure. This Frobenius manifold has dimension $r-1$ and is described by a potential $F^{\rspin}(t^1,\ldots,t^{r-1})$ with the metric~$\eta$, given by $\eta_{\alpha\beta}=\delta_{\alpha+\beta,r}$, and the Euler vector field
$$
E=\sum_{\alpha=1}^{r-1}\frac{r+1-\alpha}{r}t^\alpha\frac{\d}{\d t^\alpha},\qquad \sum_{\alpha=1}^{r-1}E^\alpha\frac{\d F^\rspin}{\d t^\alpha}=\frac{2r+2}{r}F^\rspin.
$$  
The conformal dimension is $\delta=\frac{r-2}{r}$. The potential $F^{\rspin}$ is a polynomial in the variables~$t^1,\ldots,t^{r-1}$. For more details, we refer a reader, for example, to~\cite{PPZ16,BCT17}.

The generating series of the descendent integrals with Witten's class over the moduli space of curves of genus $0$ with an $r$-spin structure gives the descendent potential $\mcF^{\rspin}(t^*_*)$, corresponding to our Frobenius potential $F^{\rspin}$. This descendent potential corresponds to a calibration with all the matrices $R_i$ being zero. Thus, the Virasoro operators $L_m$ are given by
\begin{align*}
L_m=&\sum_{\substack{1\le\alpha\le r-1\\a\ge 0}}\left(\frac{\alpha}{r}+a\right)^{\overline{m+1}}(t^\alpha_a-\delta^{\alpha,1}\delta_{a,1})\frac{\d}{\d t^\alpha_{a+m}}+\frac{\eps^2}{2}\sum_{\substack{\alpha+\beta=r\\a+b=m-1}}\left(\frac{\alpha}{r}\right)^{\overline{a+1}}\left(\frac{\beta}{r}\right)^{\overline{b+1}}\frac{\d^2}{\d t^\alpha_a\d t^\beta_b}+\\
&+\delta_{m,-1}\frac{1}{2\eps^2}\sum_{\alpha+\beta=r}t^\alpha_0t^\beta_0+\delta_{m,0}\frac{r^2-1}{24r},
\end{align*}
where we use the notation 
$$
x^{\overline{n}}:=
\begin{cases}
\prod_{i=1}^n(x+i-1),&\text{if $n\ge 1$},\\
1,&\text{if $n=0$},
\end{cases}
$$
for a complex number $x$ and an integer $n\ge 0$.

In~\cite{JKV01} the authors considered a generalization of the $r$-spin theory, that we call the extended $r$-spin theory, and in~\cite{BCT17} the authors noticed (see Remark~\ref{remark about open WDVV}) that such a generalization produces a solution $F^{\ext}(t^1,\ldots,t^r)$ of the open WDVV equations, satisfying condition~\eqref{eq:unit condition for Fo} and the homogeneity condition  
$$
\sum_{\alpha=1}^{r-1}E^\alpha\frac{\d F^{\ext}}{\d t^\alpha}+\frac{1-\delta}{2}t^r\frac{\d F^{\ext}}{\d t^r}=\frac{3-\delta}{2}F^{\ext}.
$$  
Recall that we identify $t^r=s$. Note that the function $F^{\ext}$ also controls the open $r$-spin theory, constructed in~\cite{BCT18}. The generating series $\mcF^{\ext}(t^*_*)$ of the descendent integrals in the extended $r$-spin theory is the open descendent potential, corresponding to the function~$F^{\ext}$ and a calibration with all the matrices $\tR_i$ being zero. Thus, the associated open Virasoro operators $\mcL_m$ are
$$
\mcL_m=L_m+\sum_{d\ge 0}\frac{(d+m+1)!}{d!}t^r_d\frac{\d}{\d t^r_{d+m}}+\eps\frac{3(m+1)!}{4}\frac{\d}{\d t^r_{m-1}}+\delta_{m,-1}\eps^{-1}t^r+\delta_{m,0}\frac{3}{4}.
$$ 

\subsection{Solutions given by the canonical coordinates}

Consider a conformal Frobenius manifold given by a potential $F = F(t^1,\dots,t^N)$, a metric $\eta$ and an Euler vector field $E$. Suppose that the Frobenius manifold is semisimple and let $u^1,\dots,u^N$ be the canonical coordinates. It is well-known that in the canonical coordinates the Euler vector field $E$ looks as
$$
E=\sum_{i=1}^N(u^i+a^i)\frac{\d}{\d u^i},
$$
for some constants $a^i\in\mbC$.
\begin{proposition}
For any $1\le k\le N$ the function $F^o=u^k\cdot s$ satisfies the open WDVV equations together with condition~\eqref{eq:unit condition for Fo} and the homogeneity condition
\begin{gather}\label{eq:homogeneity for canonical coordinates}
\sum_{\alpha=1}^N E^\alpha\frac{\d F^o}{\d t^\alpha}+\frac{1-\delta}{2}s\frac{\d F^o}{\d s}=\frac{3-\delta}{2}F^o+a^ks.
\end{gather}
\end{proposition}
\begin{proof}
Recall that the canonical coordinates satisfy the equations
\begin{equation}\label{eq:canonical coordinates}
c_{\alpha\beta}^\nu\frac{\p u^k}{\p t^\nu} = \frac{\p u^k}{\p t^\alpha}\frac{\p u^k}{\p t^\beta},\quad 1\le\alpha,\beta\le N,
\end{equation}
which immediately imply equations~\eqref{eq:open WDVV,2}. Equations~\eqref{eq:open WDVV,1} for the function $F^o=u\cdot s$ are equivalent to
\begin{gather}\label{eq:open WDVV1 for canonical coordinates}
c_{\alpha\beta}^\nu \frac{\p^2 u^k}{\p t^\nu \p t^\gamma}+\frac{\p^2 u^k}{\p t^\alpha \p t^\beta}\frac{\p u^k}{\p t^\gamma}=c_{\gamma\beta}^\nu \frac{\p^2 u^k}{\p t^\nu \p t^\alpha} + \frac{\p^2 u^k}{\p t^\gamma \p t^\beta}\frac{\p u^k}{\p t^\alpha}.
\end{gather}
Differentiating equation~\eqref{eq:canonical coordinates} with respect to $t^\gamma$, we get
\[
\frac{\p c_{\alpha\beta}^\nu}{\p t^\gamma} \frac{\p u^k}{\p t^\nu} = \frac{\p^2 u^k}{\p t^\gamma \p t^\alpha} \frac{\p u^k}{\p t^\beta} + \frac{\p u^k}{\p t^\alpha} \frac{\p^2 u^k}{\p t^\gamma \p t^\beta} - c_{\alpha\beta}^\nu \frac{\p^2 u^k}{\p t^\nu \p t^\gamma}.
\]
Combining this equation with the similar equation for $\dfrac{\p c_{\gamma\beta}^\nu}{\p t^\alpha} \dfrac{\p u^k}{\p t^\nu}$ and noting that $\dfrac{\p c_{\gamma\beta}^\nu}{\p t^\alpha} = \dfrac{\p c_{\alpha\beta}^\nu}{\p t^\gamma}$, we get equation~\eqref{eq:open WDVV1 for canonical coordinates}.
  
Property~\eqref{eq:unit condition for Fo} follows from the fact that $\frac{\d}{\d t^1}=\sum_{i=1}^n\frac{\d}{\d u^i}$. The homogeneity property~\eqref{eq:homogeneity for canonical coordinates} is obvious.
\end{proof}

\subsection{Open Gromov--Witten theory of $\PP^1$}

Consider the $2$-dimensional Frobenius manifold given by the Gromov--Witten theory of $\PP^1$:
\[
F(t_1,t_2) = \frac{1}{2}t_1^2t_2+e^{t_2}.
\]
The Euler vector field is $E = t_1\frac{\p}{\p t_1} + 2 \frac{\p}{\p t_2}$ and $\delta = 1$. Let us find all solutions~$F^o(t_1,t_2,s)$ of the open WDVV equations~\eqref{eq:open WDVV,1},~\eqref{eq:open WDVV,2}, satisfying condition~\eqref{eq:unit condition for Fo} and the homogeneity condition
$$
t_1\frac{\d F^o}{\d t_1}+2\frac{\d F^o}{\d t_2}=F^o+D_1 t_1+D_2 t_2+\tD s+E.
$$
We consider such solutions up to adding a constant and linear terms in the variables $t_1,t_2$ and~$s$. Then we can assume that $D_1=0$. The general form of such a function $F^o$ is $F^o = t_1 s + e^{\frac{t_2}{2}} \phi(t_2,s)$, for some function $\phi(t_2,s)$, satisfying $2e^{\frac{t_2}{2}}\frac{\d\phi}{\d t_2}=D_2 t_2+\tD s+E$. Making the transformation $\phi\mapsto\phi-(D_2(t_2+2)+\tD s+E)e^{-\frac{t_2}{2}}$, we come to a function $F^o$ of the form
\begin{gather}\label{eq:form for open P1}
F^o = t_1 s + e^{\frac{t_2}{2}} \phi(s),
\end{gather}
where the function $\phi$ depends only on $s$. Such a function $F^o$ satisfies the homogeneity property
$$
t_1\frac{\d F^o}{\d t_1}+2\frac{\d F^o}{\d t_2}=F^o
$$
and condition~\eqref{eq:unit condition for Fo}.

The system of open WDVV equations~\eqref{eq:open WDVV,1},~\eqref{eq:open WDVV,2} for a function $F^o$ of the form~\eqref{eq:form for open P1} is equivalent to the equation $(\phi')^2 - \phi\phi'' = 4$ and, solving this ordinary differential equation, we get the following two-parameter family of solutions: 
\begin{gather*}
F^o_{\alpha,\beta} = t_1s\pm 2\alpha^{-1}e^{\frac{t_2}{2}}\sinh\left(\alpha(s+\beta)\right),\quad\alpha,\beta\in\mbC.
\end{gather*}
For $\alpha=0$ we get the functions
$$
F^o_{0,\beta}=\left.\left(t_1s\pm 2\alpha^{-1}e^{\frac{t_2}{2}}\sinh\left(\alpha(s+\beta)\right)\right)\right|_{\alpha=0}=t_1s\pm 2 e^{\frac{t_2}{2}}(s+\beta),
$$
which correspond to the solutions given by the canonical coordinates $u_1=t_1+2e^{\frac{t_2}{2}}$ and $u_2=t_1-2e^{\frac{t_2}{2}}$ of our Frobenius manifold.

\begin{remark}
We believe that the function $F^o_{\alpha,\beta}$ with a correctly chosen calibration and the corresponding open descendent potential should control the genus $0$ open Gromov--Witten invariants of $\PP^1$, which don't have a rigorous geometric construction at the moment.
\end{remark}

\end{document}